\documentclass[12pt]{amsart}
\usepackage{fullpage}
\usepackage{graphicx}
\usepackage{bbm}
\usepackage{amsfonts}
\usepackage{amssymb}
\usepackage{amsmath}
\usepackage{amsthm}
\usepackage{latexsym}

\newtheorem{theorem}{Theorem}[subsection]

\newtheorem{corollary}[theorem]{Corollary}

\newtheorem{proposition}[theorem]{Proposition}
\newtheorem{definition}[theorem]{Definition}

\theoremstyle{plain}
\theoremstyle{plain}

\theoremstyle{remark}
    \newtheorem{remark}[theorem]{Remark}

\newcommand{\EE}{\mathbb{E} }

\newcommand{\one}{\mathbbm{1} }
\newcommand{\PP}{\mathbb{P} }

\newcommand{\RR}{\mathbb{R} }

\newcommand{\ff}{\mathcal{F} }

\newcommand{\BS}{C_\mathrm{BS}}
\newcommand{\IBS}{Y_\mathrm{BS}}

\newcommand{\ignore}[1]{}
\begin{document}

\author{Michael R. Tehranchi \\
University of Cambridge }
\address{Statistical Laboratory\\
Centre for Mathematical Sciences\\
Wilberforce Road\\
Cambridge CB3 0WB\\
UK}
\email{m.tehranchi@statslab.cam.ac.uk}

\date{\today}

\title{A Black--Scholes inequality: applications and generalisations}

\begin{abstract}  The space of call price
functions has a natural noncommutative semigroup structure with an involution. A basic example is the Black--Scholes
call price surface, from which an interesting   inequality for Black--Scholes
implied volatility is derived.   
The binary operation  is
compatible with the convex order, and therefore a one-parameter sub-semigroup gives rise to an arbitrage-free market model.  It is shown that each such one-parameter semigroup corresponds 
to a unique log-concave probability density, providing a family of tractable call price surface parametrisations
in the spirit of the Gatheral--Jacquier SVI surface. An explicit 
example is given to illustrate the idea. The key observation is
an isomorphism linking an initial call price curve to the lift zonoid of the
terminal price of the underlying asset.
\end{abstract}

\maketitle
\section{Introduction}\label{se:intro}
We define the Black--Scholes call price function $\BS: [0,\infty) \times [0,\infty) \to [0,1]$  by the formula
\begin{align*}
\BS(\kappa,y) &=  \int_{-\infty}^{\infty} (\varphi(z+y) - \kappa \ \varphi(z))^+ d z  \\
& = \left\{ \begin{array}{ll}
				 {\Phi}\big(-\frac{\log \kappa}{y} + \frac{y}{2}\big) - \kappa {\Phi}
				\big(-\frac{\log \kappa}{y} - \frac{y}{2}\big) & \mbox{ if } y > 0, \kappa > 0,\\
				(1-\kappa)^+ & \mbox{ if } y =0,  				\\
				1 & \mbox{ if } \kappa =0,
				\end{array} \right.
\end{align*}
where $\varphi(z) = \frac{1}{\sqrt{2\pi}} e^{-z^2/2}$ is the standard normal density and $\Phi(x) = \int_{-\infty}^x \varphi(z) d z$ is
its distribution function. 
Recall the financial context of this definition:
a market with a risk-free 
zero-coupon bond of unit face value, maturity $T$ and 
initial price $B_{0,T}$;
a  stock  with initial price $S_{0}$ that pays no dividend;
and a European call option written on the stock with  maturity $T$ and strike price $K$. 
In the Black--Scholes model, 
the initial  price $C_{0,T,K}$ of the  call option  
is given by the formula
$$
C_{0,T,K} = S_0 \   \BS\left( \frac{K B_{0,T} }{S_0 }  , \sigma \sqrt{T} \right),
$$
where    $\sigma$
is the volatility of the stock price.
 In particular,   the first argument of $\BS$ plays the role of the moneyness $\kappa=K B_{0,T}/S_0$
and the second argument  plays the role of the total standard deviation $y = \sigma \sqrt{T}$ of the terminal log stock price.

The starting point of this note is the following   observation.

\begin{theorem}\label{th:BS}  For $\kappa_1, \kappa_2 > 0$ and $y_1, y_2 > 0$   we have
$$
\BS(\kappa_1  \kappa_2, y_1 + y_2  ) \le \BS(\kappa_1, y_1  ) + \kappa_1 \BS(\kappa_2, y_2  )
$$
with equality if and only if
$$
-\frac{ \log \kappa_1}{y_1} - \frac{y_1}{2} = -\frac{   \log \kappa_2}{y_2} + \frac{y_2}{2}.
$$
\end{theorem}
While it is fairly straight-forward to prove Theorem \ref{th:BS} directly, the proof is omitted as 
it is a special case of  Theorem \ref{th:peacock} below.  
Indeed,  the purpose of this note is to try to understand the 
fundamental principle that gives rise to such an inequality.
As a hint of things to come, it is worth pointing out that   the expression $y_1+y_2$ 
appearing on the left-hand side of the inequality corresponds to the sum of the
standard deviations -- not the sum of the variances.  From this observation, 
it may not be surprising to see that a key idea underpinning Theorem \ref{th:BS} 
is that of adding comonotonic -- not independent -- normal random variables.  These
vague comments will be made precise in Theorem \ref{th:countermonotone} below.

Before proceeding, we re-express Theorem \ref{th:BS} in terms of the 
Black--Scholes implied total standard deviation function, defined for $\kappa > 0$ to be the inverse function 
$$
\IBS(\kappa, \cdot): [ (1-\kappa)^+, 1 ] \to [0,\infty]
$$
such that 
$$
y = \IBS(\kappa, c) \Leftrightarrow   \BS(\kappa,y) = c.
$$
In particular, the quantity $\IBS(\kappa,c)$ denotes the implied total standard deviation
 of an option of  moneyness $\kappa$
whose normalised   price is $c$.  
We will find it notationally
convenient to set 
$\IBS(\kappa,c) = \infty$ for $c \ge 1$.  With this notation, we
have the following interesting reformulation which requires no proof:

\begin{corollary}\label{th:ivcor}
For all $\kappa_1, \kappa_2 > 0$ and $(1-\kappa_i)^+ < c_i < 1 $ for $i=1,2$, we have
$$
\IBS(\kappa_1, c_1) + \IBS(\kappa_2, c_2) \le \IBS( \kappa_1 \kappa_2, c_1 + \kappa_1 c_2 )  
$$
with equality if and only if
$$
-\frac{ \log \kappa_1}{y_1} - \frac{y_1}{2} = -\frac{  \log \kappa_2}{y_2} + \frac{y_2}{2}.
$$
where $y_i = \IBS(\kappa_i, c_i)$ for $i=1,2$.
\end{corollary}

To add some context, we recall the following related bounds on the function $\BS$ and $\IBS$;
see \cite[Theorem 3.1]{SIFIN}.

\begin{theorem}\label{th:SIFIN}
For all $\kappa > 0$, $y > 0$, and $0 < p < 1$ we have
$$
\BS(\kappa,y) \ge \Phi( \Phi^{-1}(p) + y) - p\kappa
$$
with equality if and only if
$$
p = \Phi\left(-\frac{\log \kappa}{y} - \frac{y}{2} \right).
$$
Equivalently, for all  $\kappa > 0$,  $(1-\kappa)^+ < c < 1$ and $0 < p < 1$ we have
$$
\IBS(\kappa,c) \le \Phi^{-1}(c + p\kappa) - \Phi^{-1}(p)
$$
where $\Phi^{-1}(u) = + \infty$ for $u \ge 1$.
\end{theorem}

In \cite{SIFIN}, Theorem \ref{th:SIFIN} was used to derive upper bounds on the
implied total standard deviation function $\IBS$ by selecting various values of $p$
to insert into the inequality. 

The
function $\Phi \big( \Phi^{-1}(\cdot) + y \big)$ has appeared elsewhere in various contexts.
For instance, it is the value function for a problem of maximising the probability
of hitting a target considered by Kulldorff \cite[Theorem 6]{kulldorff}.   (Also see the book of Karatzas\cite[Section 2.6]{karatzas}.) 
In insurance mathematics, the function is often called the Wang transform
and was proposed in \cite{wang} as a method of distorting a probability distribution in order to introduce a risk premium.  In a somewhat
unrelated context, Kulik \& Tymoshkevych \cite{kt} observed, 
while proving a certain log-Sobolev inequality,  that
the family of functions
$
\left( \Phi( \Phi^{-1}(\cdot) + y) \right)_{y \ge 0}
$
forms a semigroup under function composition.  
We   will see that this semigroup property is the 
essential idea of our proof of Theorem \ref{th:BS}  and its subsequent  generalisations.

The rest of this note is arranged as follows. In section \ref{se:property} we 
introduce a space of call price curves and explore some of its properties.
In particular, we will see that it has a natural noncommutative semigroup structure with an involution.
The binary operation has a natural financial interpretation as the maximum
price of an option to swap the one asset
for a fixed number of shares of a second asset. 
In section \ref{se:peacock}, we introduce a space of call price
surfaces and provide in Theorem \ref{th:char-surf} equivalent
characterisations in terms of either one supermartingale or two martingales.  Furthermore, it is shown that the binary operation is
compatible with the decreasing convex order, and therefore a one-parameter semigroup of the space of call
curves can be associated with an arbitrage-free
market model.  A main result of this article is Theorem \ref{th:classify}: each one-parameter semigroup corresponds to a unique (up to translation and scaling) log-concave  probability density, 
generalising the Black--Scholes call price surface and providing a family of reasonably 
tractable call surface parametrisations
in the spirit of the SVI surface. 
In section \ref{se:calibrating}, further 
properties of these call price surfaces, including the asymptotics
of their implied volatility, are explored. In addition, an explicit 
example is given to illustrate the idea, and is calibrated to real world call price
data.   In section \ref{se:proofs}, the proof of Theorem \ref{th:classify}
is given. The key observation is
 that the 
Legendre transform is an isomorphism converting the binary operation on
  call price curves to function composition.  The isomorphism has
  the additional interpretation as the lift zonoid of the
terminal price of the underlying asset.

\section{The algebraic properties of call prices}\label{se:property}
 \subsection{The space of call price curves}
For motivation, consider a market with two (non-dividend paying) assets whose
prices at time $t$ are
$A_t $ and $B_t$. We assume that both prices are always 
non-negative and that the initial prices $A_0$ and $B_0$ are strictly positive.  
We further assume that there exists a martingale deflator $Y =(Y_t)_{t \ge 0}$, that is, a positive adapted
process such that the processes $YA$ and $YB$ are both martingales. 
The assumption of the existence of a martingale deflator
ensures that there is no arbitrage in the market. (Conversely,
in discrete time, no arbitrage implies the existence a martingale deflator, even
if the market does not admit a num\'{e}raire portfolio;
see \cite{no-num}.)  

Now introduce an option to swap one share of asset
$A$ with $K$ shares of asset $B$ at a fixed time $T > 0$,
so the payout is $(A_T - K B_T)^+$.   If   the asset $B$ is a risk-free zero-coupon bond of maturity $T$ and unit face value, then the 
option is a standard call option. It will prove useful in our
discussion to let asset $B$
be arbitrary, but we shall still refer to this option as a call option.

There is no
arbitrage in the augmented market if the time $t$ price of this call option 
is
$$
C_{t,T, K} = \frac{1}{Y_t}\EE [ Y_T( A_T - K B_T)^+ | \ff_t]. 
$$
In particular, setting 
$$
\alpha = \frac{Y_T A_T}{Y_0 A_0} \ \mbox{ and } \  \beta  = \frac{Y_T B_T}{Y_0 B_0}
$$ 
the initial price of this option, normalised by the 
initial price of asset $A$ can be written as
$$
\frac{C_{0,T,K}}{A_0} = \EE[ (\alpha - \kappa \beta)^+ ].
$$
where the moneyness is given by
$$
 \kappa = \frac{K B_0}{A_0}.
 $$

The above discussion motivates the following definition:
\begin{definition}\label{de:C}  
A function $C:[0,\infty) \to [0, 1]$ is a 
call price curve iff there exist non-negative random variables 
 $\alpha$ and $\beta$ defined on some probability space such that
$$
\EE(\alpha) = 1 = \EE(\beta).
$$ 
and 
$$
C(\kappa) = \EE[ (\alpha - \kappa \beta)^+ ] \mbox{ for all } \kappa \ge 0,
$$
in which case the ordered pair $(\alpha, \beta)$ of random variables is called a basic representation of $C$.  
The set
of all   call price curves is denoted $\mathcal C$.   
\end{definition}

From a practical perspective, the normalised call price  $C(\kappa)$
is directly observed, while the law of the pair $(\alpha, \beta)$
is not.
Therefore, a theme of this note is to try to express notions in
terms of the call price curve.  Here is a first result
of this type.

\begin{theorem}\label{th:rep} Given a function 
 $C:[0,\infty) \to [0,1]$, the following are equivalent:
\begin{enumerate}
\item  $C \in \mathcal C$.
\item There exists a non-negative random variable $S$   with $\EE(S) \le 1$
such that
$$
C(\kappa) =  1 - \EE[ S \wedge \kappa ] \mbox{ for all } \kappa \ge 0.
$$
\item $C$ is convex and such that 
$C(\kappa) \ge (1-\kappa)^+$  for all $\kappa \ge 0$.
\end{enumerate}
Furthermore, in case (2) we have that
$$
\PP(S > 0 ) = - C'(0)  \mbox{ and } \EE(S ) = 1 - C(\infty).
$$ 
and more generally that
 $$
 \PP(S > \kappa ) = - C'(\kappa) \mbox{ for all } \kappa \ge 0,
 $$  where
 $C'$ denotes the right-derivative of $C$.
\end{theorem}
 
\begin{proof}
The implications (1)$\Rightarrow$(3) and (2)$\Rightarrow$(3) are
straightforward, so their proofs are omitted.  Furthermore, the claim
that the distribution of $S$ can be recovered from $C$ is essentially the
Breeden \& Litzenberger \cite{BL} formula.

 (3)$\Rightarrow$(2):  By convexity, 
the right-derivative $C'$ is defined everywhere and is non-decreasing and
right-continuous. Furthermore, since $(1-\kappa)^+ \le C(\kappa) \le 1$
for all $\kappa$ we have $-1 \le C'(\kappa) \le 0$ for all $\kappa$. 
Let $S$ be a random variable such that $\PP(S > \kappa) = - C'(\kappa)$.
Note that
\begin{align*}
\EE[ S \wedge \kappa ] &= \EE \int_0^\kappa \one_{\{ u < S \}} du \\
& = \int_0^\kappa \PP(S > u) du \\
& = 1 - C(\kappa)
\end{align*}
by Fubini's theorem and the absolute continuity of the convex function  $C$.  

It remains to show that either 
(2)$\Rightarrow$(1) or 
(3)$\Rightarrow$(1).  That is, we must  
construct a basic representation $(\alpha, \beta)$ from either the random variable $S$ or the function $C$.
 We will give a construction showing (2)$\Rightarrow$(1)
in the proof of  Theorem
\ref{th:bullet-S}, and a rather 
different construction showing (3)$\Rightarrow$(1) in the proof of 
Theorem \ref{th:char-surf}. To avoid repetition, we omit a construction here.
\end{proof}

By definition, a call price curve $C$ is determined by two
random variables   $\alpha$ and $\beta$.   However, the
distribution of the pair $(\alpha, \beta)$ cannot be inferred solely
from $C$.  
In contrast, Theorem \ref{th:rep} above says 
that a call price curve $C$ is also determined 
by a \textit{single} random variable  $S$, and furthermore, 
the law of $S$ is \textit{unique} and  can be recovered from $C$. This observation
motivates the following definition.

\begin{definition} Given a call price curve $C \in \mathcal C$, suppose that $S$ is a non-negative random 
variable such that
 $C(\kappa) = 1 - \EE[ S \wedge \kappa]$ for all $\kappa \ge 0$. 
 Then $S$ is called a primal representation of   $C$.
 \end{definition}
 
 \begin{remark}  As hinted by the name \textit{primal}, we will
 shortly introduce a \textit{dual} representation.
 \end{remark}

Figure \ref{fi:C} plots the graph of a typical element $C \in \mathcal C$.

\begin{figure} 
\caption{The graph of a typical function $C \in \mathcal C$}
	  \label{fi:C}
		\includegraphics[trim = 0cm 0 0cm 0, clip, scale = 0.40]{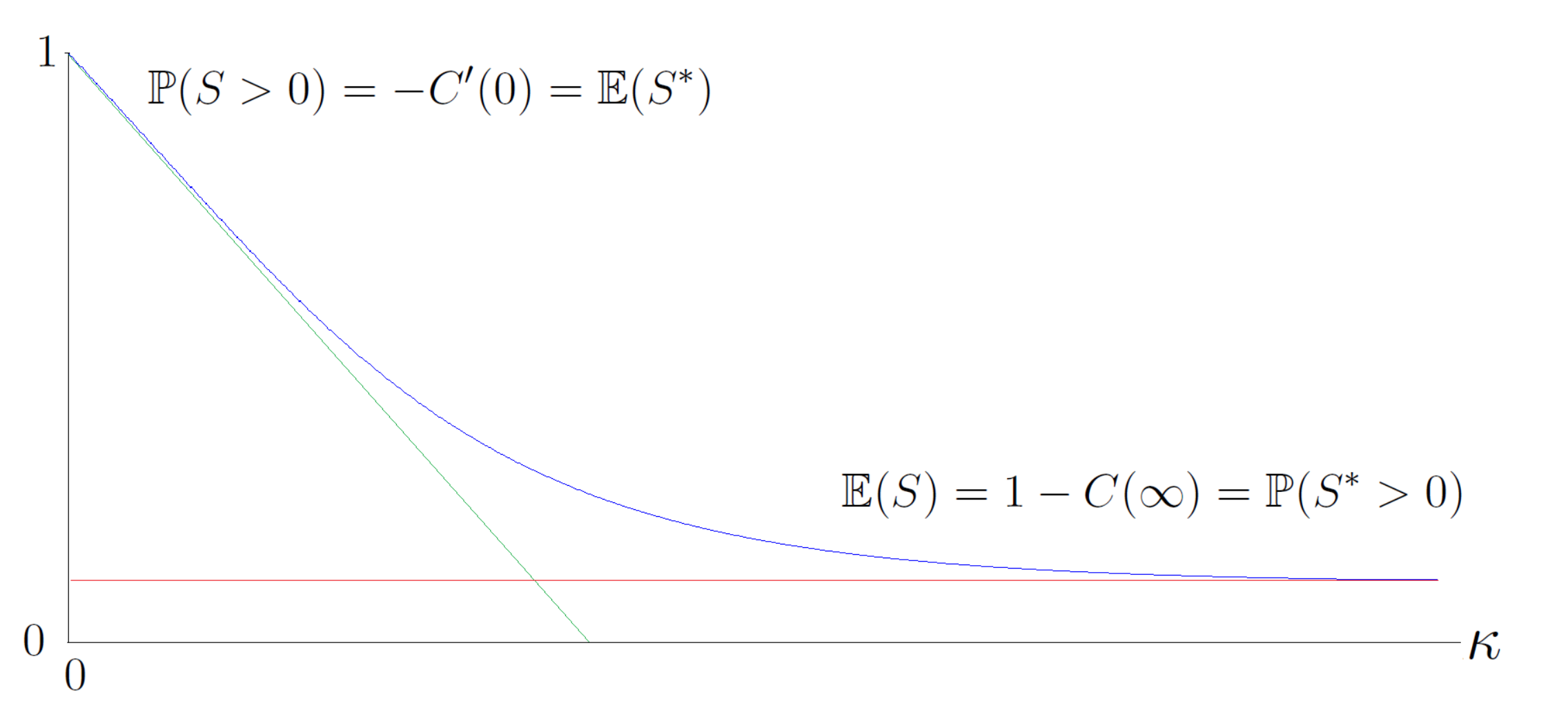}		
	\end{figure}

\begin{remark} 
An example of an element of $\mathcal C$ is the Black--Scholes call price function
$\BS(   \cdot , y )$ for any $y \ge 0$.   A primal representation is 
$$
S^{(y)} = \frac{\varphi(Z+y)}{\varphi(Z)} =  e^{-y Z - y^2/2}
$$
where $Z \sim N(0,1)$ has the standard normal distribution.
\end{remark}

\begin{remark}
We note that there are alternative financial interpretations
of call price curves  $C \in \mathcal C$ in the case $C(\infty) > 0$.
One popularised by  Cox \& Hobson \cite{CH} is to model the
primal representation as the terminal price $S$ of an asset
experiencing a bubble in the sense that the 
price process discounted by the price of the risk-free $T$-zero coupon
bond is a strictly local martingale under a fixed $T$-forward measure.
For this interpretation, the option payout must be modified: rather
than the payout of standard (naked) call option, in this interpretation
the quantity $C(\kappa)$ models the normalised price of a  fully collateralised (covered) call 
with payout   
$$
(S - \kappa)^+ + 1 - S = 1 - S\wedge  \kappa.
$$  

In my view, there are two related shortcomings of this interpretation. 
Firstly, this type of bubble phenomenon can only arise in continuous time models, since in discrete time non-negative local martingales are
necessarily true martingales.  Secondly, in the case  $\EE(S) < 1$ where the underlying stock is not priced by expectation, it is not clear from
a modelling perspective why the market should then price the call option by expectation $C(\kappa) = \EE[ (S - \kappa)^+ + 1 - S].$  Both shortcomings
highlight the subtlety of continuous time arbitrage theory, in particular,
the sensitive dependence on the choice of num\'eraire on the definition of 
arbitrage (and related arbitrage-like conditions).
\end{remark}


\subsection{The involution}  

There is a natural involution on the
space of call prices:

\begin{definition}  
Given a call price curve $C \in \mathcal C$ with basic representation
$(\alpha, \beta)$, the function $C^*$ is the call price curve with 
basic representation $(\beta, \alpha)$.
\end{definition}

This leads to a straightforward financial interpretation of the involution. 
As described above, we may think of $C(\kappa)$ as the initial price,
normalised by $A_0$,  
of the option to swap one share of asset $A$ for $K$ shares of asset $B$, where
$K B_0 = \kappa A_0$. Then $C^*(\kappa)$ is the initial   price, normalised by
$B_0$, of the option to swap one share of asset $B$ for $K^*$ shares of asset $A$, where
$K^* A_0 = \kappa B_0$.

 We now record a fact about this involution ${ }^*$, expressed 
 directly in terms of call prices.  The proof is a straightforward
 verification, and hence omitted.

\begin{theorem}\label{th:involution}
Fix $C \in \mathcal C$.  Then
$C^*(0) = 1$ and
$$
C^*(\kappa) = 1 - \kappa + \kappa C(1/\kappa) \mbox{ for all } \kappa > 0.
$$
\end{theorem}

\begin{remark}
As an example, notice for the Black--Scholes call function 
we have
$$
 \BS(   \cdot, y)^* =  \BS(   \cdot, y) \mbox{ for all } y \ge 0
$$ 
 by the classical put-call symmetry formula.
\end{remark}

\begin{remark}
The function $C^*$ is  related to the well-known perspective function of the convex function $C$ 
defined by $(\eta, \kappa) \mapsto \eta \ C( \kappa/\eta)$;
see, for instance, the book of  Boyd \& Vanderberghe \cite[Section 3.2.6]{BoydVanderberghe}.
\end{remark}

As hinted at above, we can define another  random variable in terms
of this involution:

\begin{definition}  Given a call price curve $C \in \mathcal C$, 
a non-negative random variable $S^*$ is a dual
representation of $C$ iff $S^*$ is a primal representation of the call
price curve $C^*$. 
\end{definition}

That this dual random variable  should be called a representation of a call price is 
due to the following observation.  Again the proof is straightforward
and hence omitted.

\begin{theorem}\label{th:dual} Given a call price 
$C \in \mathcal C$ with dual representation $S^*$ we have
$$
C(\kappa) = \EE[ (1 - S^* \kappa)^+ ] = 1 - \EE[ 1 \wedge (S^* \kappa)] \mbox{ for all } \kappa \ge 0.
$$
In particular, we have
$$
\PP(S^* > 0) = 1- C(\infty) \mbox{ and } \EE(S^*) = - C'(0).
$$
Finally, for all $\kappa \ge 0$ we have
\begin{align*}
C(\kappa ) &= \PP(S^* < 1/\kappa ) - \kappa  \PP(S > \kappa ) \\
& = \PP(S^* \le 1/\kappa ) - \kappa  \PP(S \ge \kappa ).
\end{align*}
\end{theorem}


\begin{remark}   See the papers of De Marco, Hillairet \&  Jacquier
\cite{DM-H-J} and Jacquier \& Keller-Ressel \cite{J-KR} for a related
financial interpretation of the relationship between the 
primal and dual representations in terms of a continuous time market
possibly experiencing a bubble \`a la Cox \& Hobson.
\end{remark}

\subsection{The binary operation}\label{se:binary} 
We have introduced one algebraic operation, the involution ${}^*$,
to the set of call price curves.  We now come to the second algebraic operation which will help to 
contextualise the  Black--Scholes inequality of Theorem \ref{th:BS}.
To motivate it, consider a market with three assets with time $t$ prices
$A_{1,t}, A_{2,t}$ and $B_t$.   We know the initial cost of an 
option to swap one share of asset $A_1$ with $H_1$ shares of asset $B$,
as well as the initial cost of an 
option to swap one share of asset $B$ with $H_2$ shares of asset $A_2$, for
various values of $H_1$ and $H_2$, where all of the options mature
at a fixed date $T>0$.   Our goal is the find an upper bound
on the cost of an option to swap one share of asset $A_1$ for $K$ shares
of asset $A_2$, for the same maturity date $T$.

\begin{definition}
For call price curves $C_1, C_2 \in \mathcal C$, define a binary operation $\bullet$ on $\mathcal C$  by
$$
C_1  \bullet C_2 (\kappa ) = \sup_{\alpha_1, \beta, \alpha_2 }
\EE[ (\alpha_1 - \kappa \alpha_2)^+ ]
$$
where the supremum is taken over non-negative random variables $\alpha_1,
 \beta,  \alpha_2$ defined on the same probability space such that $(\alpha_1, \beta)$ is a basic representation of $C_1$ and $(\beta, \alpha_2)$
is a basic representation of $C_2$.
\end{definition}

At this stage, it is not immediately clear that given two call price curves 
$C_1$ and $C_2$ one can find 
a triple $(\alpha_1, \beta, \alpha_2)$ satisfying the definition
of the binary operation $\bullet$, and in principle, we should 
complete the definition with the usual convention that $\sup \emptyset = - \infty$.  Fortunately, this caveat is not necessary as can be
deduced from the following result:

\begin{theorem}\label{th:bullet-S}
For call price curves $C_1, C_2 \in \mathcal C$ we have
$$
C_1  \bullet C_2 (\kappa ) = 
 \sup_{S_1, S_2^*} \left\{
1 - \EE[ S_1 \wedge (S_2^* \kappa)] \right\},
$$
where the supremum is taken over random variables $S_1$ and $S_2^*$ defined
on the same space, where $S_1$ is a primal representation of $C_1$
and $S_2^*$ is a dual representation of $C_2$.
\end{theorem}

\begin{proof} First, let $S_1$ be a primal representation of $C_1$
and $S_2^*$ be a dual representation of $C_2$, defined on the same probability
space.  We will exhibit random variables $(\alpha_1, \beta, \alpha_2)$ such that $(\alpha_1, \beta)$ is a basic representation of $C_1$ and 
$(\beta, \alpha_2)$ is a basic representation of $C_2$.

For the construction, we introduce Bernoulli random variables $\gamma_1, 
\gamma_2, \delta_1, \delta_2$, independent of $(S_1, S_2^*)$ and each other,
with  
$$
\PP(\gamma_1=1) = \EE(S_1), \ \PP(\gamma_2 = 1) = \EE(S_2^*)
\mbox{ and } \PP(\delta_1 = 1) =  \PP(\delta_2 = 1) = \tfrac{1}{2}.
$$ 

If
$\PP(S_1= 0) < 1$ then set 
$$
a_1 = \frac{S_1}{\EE(S_1)} \mbox{ and } b_1 = \frac{\gamma_1}{\EE(S_1)}
$$
and if $S_1=0$ almost surely, set $a_1= 2 \delta_1$ and $b_1=2(1-\delta_1)$.
Similarly, if 
$\PP(S_2^*= 0) < 1$ then set 
$$
a_2 = \frac{S_2^*}{\EE(S_2^*)} \mbox{ and } b_2 = \frac{\gamma_2}{\EE(S_2^*)}
$$
and if $S_2^*=0$ almost surely, set $a_2= 2 \delta_2$ and $b_2=2(1-\delta_2)$.
Finally set
$$
\alpha_1 = a_1 b_2, \ \ \beta = b_1 b_2, \ \ \alpha_2 = a_2 b_1.
$$
It is easy to check that the triplet $(\alpha_1, \beta, \alpha_2)$ is
the desired representation.  This shows
$$
C_1  \bullet C_2 (\kappa ) \ge
 \sup_{S_1, S_2^*} \left\{
1 - \EE[ S_1 \wedge (S_2^* \kappa)] \right\},
$$

For the reverse inequality,  given a basic representation  $(\alpha_1, \beta)$   of $C_1$ and 
  a basic representation $(\beta, \alpha_2)$ of $C_2$ defined on the same 
  probability space $(\Omega, \ff, \PP)$, we let
  $$
  S_1 = \frac{\alpha_1}{\beta} \one_{\{\beta > 0  \} }
  \mbox{ and }  S_2^* = \frac{\alpha_2}{\beta} \one_{\{\beta > 0  \} } 
$$
and an absolutely continuous measure $\PP^\beta$ by
$
\frac{d \PP^\beta}{d \PP} = \beta.
$
It is easy to check that $S_1$ is a primal representation of $C_1$
and $S_2^*$ is a dual representation of $C_2$ under $\PP^\beta$ and that
\begin{align*}
\EE[ (\alpha_1 - \kappa \alpha_2)^+ ] & \le 1 - \EE  [ \alpha_1 \wedge (\alpha_2 \kappa) \one_{\{\beta > 0  \}}] \\
& = 1 - \EE^{\beta} [ S_1 \wedge (S_2^* \kappa)].
\end{align*}
\end{proof}

Given the laws of two random variables $X_1$ and $X_2$ and a convex
function $g$, it is well-known that the quantity
$$
\EE[ g(X_1 + X_2) ]
$$
is maximised when $X_1$ and $X_2$ are comonotonic.   See, for instance,
the paper of Kaas--Dhaene--Vyncke--Goovaerts-- Denuit \cite{KDVGD} for a proof.
By rewriting the expression
$$
1 -  S_1 \wedge (S_2^* \kappa) = (S_1 - \kappa S_2^*)^+  + 1 - S_1
$$ 
we see that the supremum defining the 
binary operation $\bullet$ is achieved 
when $S_1$ and $S^*_2$ are countermonotonic. 
We will recover this fact in the following result,
which also continues our theme of expressing notions directly
in terms of the call prices. In this case, the binary operation $\bullet$
can be expressed via a minimisation problem:

\begin{theorem} 
\label{th:countermonotone}
Let $S_1$ be a primal representation of $C_1 \in \mathcal C$, and $S_2^*$ a dual
representation of $C_2 \in \mathcal C$, where  $S_1$ and $S_2^*$ are defined on the same probability space. 
Then
$$
1 - \EE[ S_1 \wedge (\kappa S_2^*)] 
\le C_1(\eta) + \eta \ C_2(\kappa/\eta)
$$
for all $\kappa\ge 0$ and $\eta \ge 0$, with convention $0 \ C_2(\kappa/0) = 0$.  There is equality if   the following
hold true: 
\begin{enumerate}
\item $S_1$ and $S_2^*$ are countermonotonic, and
\item $\PP(S_1 < \eta) \le  \PP(S_2^* \ge \eta/\kappa) $ and  $
\PP(S_1 \le \eta)  \ge \PP(S_2^* > \eta/\kappa).$
\end{enumerate}

In particular, we  have
$$
C_1 \bullet C_2(\kappa) = \inf_{\eta \ge 0} [C_1(\eta) + \eta C_2(\kappa/\eta)]
\mbox{ for all } \kappa \ge 0.
$$
\end{theorem}

\begin{proof}  
Recall  that for real $a,b$ we have
$$
(a+b)^+ \le a^+ + b^+
$$
with equality if   $ab \ge 0$.  Hence, fixing $\kappa \ge 0$,
 we have
\begin{align*}
1 - \EE[ S_1  \wedge (S_2^* \kappa) ] & = \EE[ (S_1 - \kappa S^*_2)^+ ] + 1 - \EE(S_1) \\
& \le \EE[ (S_1 - \eta)^+] +    1 - \EE(S_1) + \EE[ ( \eta - \kappa S^*_2)^+ ] \\& =  C_1(\eta) + \eta C_2(\kappa/\eta). 
\end{align*} 
for all $\eta \ge 0$.

Now pick $\eta \ge 0$ such that  
$$
\PP(S_1 < \eta) \le  \PP(S_2^* \ge \eta/\kappa) 
$$
and
$$
\PP(S_1 \le \eta)  \ge \PP(S_2^* > \eta/\kappa).
$$
Also assume that   $S_1$ and $S_2^*$ are countermonotonic so that
$$
\{S_1 < \eta \} \subseteq \{ S_2^* \ge \eta/\kappa \} 
$$
and
$$
\{ S_1 \le \eta \} \supseteq \{ S_2^* > \eta/\kappa \}.
$$
Notice that in this case, we have
$$
(S_1 - \eta)( \eta - \kappa S^*_2) \ge 0 \mbox{ almost surely}
$$
and hence there is equality in the inequality above.
\end{proof}

\begin{remark}
This  result 
is related to the upper bound on basket options found
by Hobson, Laurence \& Wang \cite[Theorem 3.1]{HLW}.  
\end{remark}

\begin{remark}
Given the conclusion of Theorem \ref{th:countermonotone} 
we caution that the operation
$\bullet$ is not the well-known inf-convolution $\square$;    however, we will see in section \ref{se:inf-conv} below
that $\bullet$ is related to the inf-convolution $\square$ via an exponential map.
\end{remark}

In light of the formula for the binary operation $\bullet$ appearing
in Theorem  \ref{th:countermonotone}, 
the Black--Scholes inequality of  Theorem \ref{th:BS} 
amounts to
the claim that for $y_1, y_2 \ge 0$ we have
$$
\BS(  \cdot, y_1 ) \bullet \BS(  \cdot, y_2 ) = \BS(   \cdot, y_1+y_2 ).
$$ 
This is a special case of Theorem \ref{th:peacock}, stated and proven below.

We now come to the key observation of this note.  
To state it, we distinguish two particular elements $E, Z \in \mathcal C$ defined by 
$$
E(\kappa) = (1-\kappa)^+ \mbox{ and } Z(\kappa) = 1 \mbox{ for all } \kappa \ge 0.
$$
Note that the random variables representing $E$ and $Z$ are constant,
with $S = 1= S^*$ representing $E$ and $S=0=S^*$ representing $Z$.
The following result shows that $\mathcal C$ is a noncommutative
semigroup with respect to  $\bullet$ with involution ${}^*$, where $E$ is the identity element $Z$ is the absorbing element.  The proof
is straightforward, and hence omitted.

\begin{theorem}\label{th:semigroup} For every $C, C_1, C_2, C_3 \in \mathcal C$ we have
\begin{enumerate}
\item   $
C_1  \bullet C_2 \in \mathcal C.
$
\item   $
C_1 \bullet  (C_2  \bullet  C_3) = (C_1  \bullet   C_2 )  \bullet  C_3.
$
\item
$(C_1 \bullet C_2)^* = C_2^* \bullet C_1^*$.
\item  $
E \bullet C = C  \bullet E = C.
$
\item
$Z \bullet C = C \bullet Z = Z.$
\end{enumerate}
\end{theorem}

We conclude this section by introducing two useful subsets of the set
of call price curves.

\begin{definition} Let
$$
\mathcal C_+ = \{ C \in \mathcal C: C'(0) = - 1 \}.
$$
and
$$
\mathcal C_1 = \{ C \in \mathcal C: C(\infty) =  0  \}.
$$
That is, fix a call price curve $C \in \mathcal C$  with primal
representation $S$ and dual representation $S^*$.
The call price curve $C$ is in $ \mathcal C_+$ if and only if 
$\PP(S > 0 ) = \EE(S^*) =  1$, while    $C$ is in $\mathcal C_1$  
if and only if $\PP(S^* > 0 ) = \EE(S) =  1$. 
\end{definition}

\begin{remark}
 As an example, notice that for the Black--Scholes call 
function we have
$$
\BS(   \cdot, y) \in C_1 \cap C_+ \mbox{ for all } y \ge 0.
$$
 \end{remark}
 
 The subsets $\mathcal C_+$ and $\mathcal C_1$ are closed
with respect to the binary operation.
\begin{proposition}\label{th:closed}
Given $  C_1, C_2 \in \mathcal C$ we have 
\begin{enumerate}
\item  $C_1 \bullet C_2 \in \mathcal C_1$ if and only
if both $C_1 \in \mathcal C_1$ and $C_2 \in \mathcal C_1$.
\item $C_1 \bullet C_2 \in \mathcal C_+$ if and only
if both $C_1 \in \mathcal C_+$ and $C_2 \in \mathcal C_+$.
\end{enumerate}
\end{proposition}

\begin{proof}  By Theorem 
\ref{th:countermonotone} we have
$$
C_1 \bullet C_2(\kappa) = 1 - \EE[ S_1 \wedge (\kappa S_2^*)] \mbox{ for all }
\kappa \ge 0
$$
where $S_1$ is a primal representation of $C_1$, where
$S_2^*$ is a dual representation of $C_2$ and $S_1$ and $S_2^*$ are
countermonotonic.  For implication (1) note that 
$$
\EE[ S_1 \wedge (\kappa S_2^*)] \to 1 \mbox{ as } \kappa \to \infty
$$
if and only if 
$$
\EE(S_1) = 1 \mbox{ and } \PP(S_2^* > 0) = 1.
$$
For implication (2), apply Theorem \ref{th:semigroup} (3) and the fact that
$(\mathcal C_+)^* = \mathcal C_1$.
\end{proof}

\section{One-parameter semigroups, peacocks and lyrebirds}\label{se:peacock}
\subsection{The space of call price surfaces}\label{se:surface}
With the motivation at the beginning of section \ref{se:property} 
we consider the family of prices of options when the maturity 
date is allowed to vary.   We now introduce the following definition:

\begin{definition} A call price surface is a function $C:[0,\infty) 
\times [0,\infty) \to [0,1]$ such that there exists an
 pair of non-negative martingales
$(\alpha_t, \beta_t)_{t \ge 0}$ such that 
$$
\alpha_0 = 1 = \beta_0
$$
and 
$$
C(\kappa, t) = \EE[ (\alpha_t - \kappa \beta_t)^+ ] \mbox{
for all } \kappa \ge 0, t \ge 0.
$$
\end{definition}

Our goal is to understand the structure the space of call price surfaces, and relate this structure
to the binary operation $\bullet$ introduced in the last section.

\begin{theorem}\label{th:char-surf}
 Given a function $C:[0,\infty)\times[0,\infty) \to [0,1]$
the following are equivalent:
\begin{enumerate}
\item $C$ is a call price surface
\item There exists a non-negative supermartingale $S$ such that $S_0=1$
and  
$$
C(\kappa, t) = 1 - \EE[ S_t \wedge \kappa] \mbox{ for all } (\kappa, t).
$$
\item There exists a non-negative supermartingale $S^*$ such that $S^*_0=1$
and  
$$
C(\kappa, t) = 1 - \EE[ 1 \wedge (\kappa S^*_t)] \mbox{ for all }  (\kappa, t).
$$
\item For all $\varepsilon > 0$, there exist bounded non-negative martingales $\alpha$ and $\beta$ such that
$\alpha_0=1 =\beta_0$ and  
$$
C(\kappa, t) = \EE[ (\alpha_t - \kappa \beta_t)^+] \mbox{ for all }  (\kappa, t)
$$
and such that  $ \alpha_t + \varepsilon \beta_t = 1 + \varepsilon$  for all $t \ge 0$.  
\item $C(\kappa, \cdot)$ is non-decreasing with $C(\kappa, 0) = (1-\kappa)^+$
for all $\kappa \ge 0$, and  $C(\cdot, t)$ is convex for all $t \ge 0$. 
\end{enumerate} 
\end{theorem}

\begin{proof} The implications ($n$) $\Rightarrow$ (5) for
$1 \le n \le 4$ are easy to check by the conditional version of Jensen's inequality.

The implications  (5) $\Rightarrow$ (2) and (5) $\Rightarrow$ (3) are proven as follows.  By Theorems \ref{th:rep} and \ref{th:dual} there
 exist families of random variables $(S_t)_{t\ge 0}$ and
 $(S^*_t)_{t \ge 0}$ such that
 $$
 C(\kappa, t) = 1 - \EE[ S_t \wedge \kappa] = 1 - \EE[ 1 \wedge (\kappa S_t^*) ]
 $$
 for all $\kappa \ge 0$ and $t \ge 0$.  The assumption that
 $C(\kappa, \cdot)$ is non-decreasing implies that both families
 of random variables (or more precisely, both families of laws) are
 non-decreasing in the decreasing-convex order.  The implications
 then follow from Kellerer' theorem \cite{kellerer}.

Implication (4) $\Rightarrow$ (1) is obvious.  So it remains to show
the implication (5) $\Rightarrow$ (4).   Fix $\varepsilon > 0$ and let
$$
\tilde C(\kappa, t) = \left\{ \begin{array}{ll}
	C\left( \frac{\varepsilon \kappa}{1+\varepsilon-\kappa}, t \right)
	 \left(1 - \frac{\kappa}{1+\varepsilon} \right) & \mbox{ if }  0 \le \kappa < 1+ \varepsilon \\
	0 & \mbox{ if } \kappa \ge  1+ \varepsilon
	\end{array} \right.
$$
It is straightforward to verify that $\tilde C$ satisfies hypothesis (5).
Hence there exists a non-negative supermartingale $\alpha$ such that
$$
\tilde C(\kappa, t) = 1 - \EE[ \alpha_t \wedge \kappa ] \mbox{ for all } (\kappa,t).
$$
But since $\tilde C(\kappa, t) = 0$ for all $\kappa \ge 1+ \varepsilon$
we can conclude that for all $t$ we have 
both $\EE( \alpha_t ) = 1$ and $ \alpha_t \le 1+ \varepsilon
$ a.s.  In particular, $\alpha$ is a true martingale 
so that
$$
\tilde C(\kappa, t) = \EE [  ( \alpha_t - \kappa )^+  ]
$$
or equivalently
$$
C(\kappa, t) = \EE [  ( \alpha_t - \kappa \beta_t)^+  ]
$$
where $\beta = \frac{1}{\varepsilon}(1+ \varepsilon - \alpha)$ as claimed.
\end{proof}

\begin{remark} The implication  (5) $\Rightarrow$ (2) is
well-known, especially in the case where $C(\infty, t) = 0$ for all $t \ge 0$
where the supermartingale $S$ is a martingale.
See, for instance, the paper of Carr \& Madan \cite{CarrMadan}.
However, implication  (5) $\Rightarrow$ (4) seems new.
\end{remark}

\subsection{One-parameter semigroups}\label{se:one-parameter}
 Returning to the topics of Section \ref{se:property}, 
 we   note that the operation $\bullet$ interacts well with the
 natural partial ordering on the space of call price curves:
\begin{proposition}\label{th:increasing}
For any $C_1, C_2 \in \mathcal C$, we have 
$$
\max\{ C_1(\kappa), C_2(\kappa) \} \le  C_1 \bullet C_2(\kappa) \mbox{
for all } \kappa \ge 0.
$$ 
\end{proposition}

\begin{proof}  Let $S_1$ be a primal representation of $C_1$ and 
$S_2^*$ a dual representation of $C_2$.  Suppose $S_1$ and $S_2^*$ are
independent. Then by Theorem \ref{th:bullet-S} we have
\begin{align*}
C_1 \bullet C_2(\kappa) & \ge 1 - \EE[ S_1 \wedge (\kappa S_2^*)] \\
& \ge 1 -   \EE[ \EE(S_1) \wedge (\kappa S_2^*)] \\
& \ge 1 - \EE[ 1 \wedge (\kappa S_2^*)] \\
& = C_2(\kappa)
\end{align*}
by first conditioning on $S_2^*$ and applying the conditional Jensen
inequality, and then using the bound $\EE(S_1) \le 1$.  The other
implication is proven similarly.
\end{proof}

Combining Theorem \ref{th:char-surf}  and Proposition \ref{th:increasing},
brings us to the main observation of this paper: 
 if $(C(\cdot, t) )_{t  \ge 0}$ is a one-parameter
sub-semigroup of $\mathcal C$  then $C(\cdot, \cdot)$ is a call price surface.
Fortunately, we will see that all such sub-semigroups can be explicitly characterised and are reasonably tractable.  

With the motivation of finding tractable family of call price
surfaces, we now study the  family of sub-semigroups of $\mathcal C$ indexed by a single parameter $y \ge 0$.   We change notation from $t$ to $y$,
since the $y$ will correspond to total implied standard deviation in
the Black--Scholes framework, so $y= \sigma \sqrt{t}$.  In particular,
we will think of $y$ not literally as the maturity date of the option, but
rather an increasing function of that date.

We will make use of the following notation.  For a probability density function $f$,  
let
$$
C_f(\kappa,y) = \int_{-\infty}^{\infty} ( f(z+y)-  \kappa f(z))^+ dz =     1 - \int_{-\infty}^{\infty}   f(z+y) \wedge  [ \kappa f(z) ] dz 
$$
for $y \in \RR$ and $\kappa \ge 0$.  Note that
$$
\BS(\cdot, y) = C_{\varphi}(\cdot, y)
$$
for $y \ge 0$, where $\varphi$ is the standard normal density.

It will be useful to distinguish a special class of densities:
\begin{definition} A probability density function $f: \RR \to [0,\infty)$ is log-concave iff
$\log f: \RR \to [-\infty, \infty)$ is concave.   
\end{definition}

We will use repeatedly a useful characterisation of log-concave densities
due to Bobkov \cite[Proposition A.1]{bobkov}:

\begin{proposition}[Bobkov]\label{th:bobkov}
Let $f$ be a probability density, with $f > 0$ on the interval $(L,R)$.  
Let $F(x) = \int_{L}^x f(z) dz$ be the corresponding 
cumulative distribution function, and $F^{-1}: [0,1] \to [L,R]$ its
quantile function.
The following are equivalent:
\begin{enumerate}
\item $f$ is log-concave.
\item $ F( F^{-1}(\cdot) + y)$ is concave on $(0,1)$ for each $y \ge 0$.
\item $ f \circ F^{-1}(\cdot)$ is concave on $(0,1)$.
\end{enumerate} 
\end{proposition}
 
Let  $f$ be a log-concave density
supported on $[L,R]$ where $- \infty \le L < R \le + \infty$.    Recall that log-concavity implies
that $f$ is continuous on the open interval $(L,R)$, but may have discontinuities at the end points.  However, without any loss of generality, we will assume throughout that $f$ is continuous on $[L,R]$.

We now present a family of one-parameter sub-semigroups of $\mathcal C$.
\begin{theorem}\label{th:peacock}
Let $f$ be a log-concave probability density function.  Then  
$$
C_f(\cdot, y_1  ) \bullet C_f( \cdot, y_2 ) = C_f(\cdot, y_1+y_2  ) \mbox{ for all } y_1, y_2  \ge 0.
$$
\end{theorem}

Note that Theorems \ref{th:countermonotone} and \ref{th:peacock}  
together says for all $\kappa_1, \kappa_2 > 0$ and $y_1, y_2 > 0$, that
$$
C_f(\kappa_1 \kappa_2, y_1 + y_2) \le C_f(\kappa_1, y_1) + \kappa_1 C_f(\kappa_2, y_2), 
$$
proving Theorem \ref{th:BS}.

While Theorem \ref{th:peacock} is not especially difficult to prove,
we will offer two proofs with each highlighting a different perspective on the
operation $\bullet$. The first is  below and the second is in Section \ref{se:proofs}.

\begin{proof}
Letting $Z$ be a random variable with density $f$, note that $f(Z+y)/f(Z)$ is a primal
representation of $C_f( \cdot, y)$.  Note also that by log-concavity of $f$, when $y \ge 0$ the function 
$z \mapsto f(z+y)/f(z)$ is non-increasing.    Similarly, 
$f(Z-y)/f(Z)$ is a dual
representation of $C_f( \cdot, y)$ and $z \mapsto f(z-y)/f(z)$ is non-decreasing. 
In particular, the random variables $f(Z+y_1)/f(Z)$ and $f(Z-y_2)/f(Z)$
are countermonotonic, and hence by Theorem \ref{th:countermonotone} we 
have
$$
C_f(\cdot, y_1) \bullet C_f( \cdot, y_2)(\kappa)
= 1 - \int_{-\infty}^{\infty} f(z+y_1) \wedge [\kappa f(z-y_2)] dz.
$$
The conclusion follows from changing variables in the integral on the right-hand side.
\end{proof}

The upshot of Theorem \ref{th:peacock} and Proposition
\ref{th:increasing} is that, given  
a log-concave density $f$, the function
 $C_f(\kappa, \cdot)$ is non-decreasing for each $\kappa \ge 0$.
Hence, given an increasing function $\Upsilon$, we can conclude from Theorem \ref{th:char-surf}  that we can
define a call price surface by
$$
(\kappa,t) \mapsto   C_f(\kappa, \Upsilon(t) ).
$$
The above formula is reasonably tractable, and
could be  seen to be in the same spirit as the
SVI parametrisation of the implied volatility surface given by Gatheral \& Jacquier \cite{GatheralJacquier}. 
Note that we can recover the Black--Scholes model by setting the density to 
$f = \varphi$ the standard normal density and the increasing function to
$
\Upsilon(t) = \sigma \sqrt{t}
$
where $\sigma$ is the volatility of the stock.
We provide another worked example in section \ref{se:examples}.

At this point we explain the name of this section.
We recall the definitions of   terms
 popularised by Hirsh, Profeta, Roynette \& Yor \cite{HPRY} and
 Ewald \& Yor \cite{ewald} among others: 
\begin{definition}
A lyrebird is a family $X = (X_t)_{t \ge 0}$ of integrable random variables
such that there exists a submartingale  $Y = (Y_t)_{t \ge 0}$ defined
on some other probability space such that $X_t \sim Y_t$ for all $t \ge 0$.
A peacock $X$ is a family of random variables such that 
both $X$ and $-X$ are lyrebirds; i.e. there exists a martingale
with the same marginal laws as $X$.
\end{definition}
The term peacock is derived from the French acronym PCOC, 
Processus Croissant pour l'Ordre Convexe, and lyrebird is
the name of an Australian bird with peacock-like tail feathers.  

 Combining Proposition   \ref{th:increasing}    and Theorem \ref{th:peacock}    
yields the following tractable family of lyrebirds and peacocks.
\begin{theorem}\label{th:surface}
Let $f$ be a log-concave density, and 
let be a random variable $Z$ have density $f$ and    let
$\Upsilon:[0,\infty) \to [0,\infty)$ be increasing.   Set
$$
S_t = \frac{f(Z+\Upsilon(t))}{f(Z)} \mbox{ for } t \ge 0.
$$
The family of random variables $-S=(-S_t)_{t \ge 0}$ is a lyrebird.
If  the support of $f$ is of the form $(-\infty, R]$, then $S$ is 
a peacock.
\end{theorem}

Note that the the semigroup $( C_f(\cdot, y) )_{y \ge 0}$ 
does not correspond
to a unique log-concave density.  Indeed,
fix a log-concave density $f$ and set
$$
f^{(\lambda,\mu)}(z) = |\lambda| \ f(\lambda z + \mu)
$$
for $\lambda, \mu \in \RR$, $\lambda \ne 0$.
Note that 
$$
C_{f^{(\lambda,\mu)}}( \kappa, y) = C_f(  \kappa, \lambda y ) \mbox{ for all } \kappa \ge 0,
y \in \RR.
$$
However, we will see below that the semigroup does identify the density $f$ 
 up to arbitrary scaling and centring parameters.  
 
 Also,
 note that  by varying the scale parameter $\lambda$ we can interpolate  between two possibilities. On the one hand, we have for all $\kappa \ge 0$
 and $y \in \RR$ that
 $$
C_{f^{(\lambda,\mu)}}(\kappa, y) \to (1- \kappa)^+  \mbox{ as }  \lambda  \to 0
$$
and on the other hand, when $y \ne 0$ that
 $$
C_{f^{(\lambda,\mu)}}(\kappa, y) \to 1  \mbox{ as } |\lambda| \to \infty
$$
by the dominated convergence theorem.

Recall that the call price curve $E(\kappa) = (1-\kappa)^+$ is the identity element for binary
operation $\bullet$.  Hence the family $C_{\mathrm{triv}}$ defined by
$C_{\mathrm{triv}}(\cdot, y) = E$
for all $y \ge 0$ is another example of a subsemigroup of $\mathcal C$. 

 Similarly, 
the call price curve $Z(\kappa) = 1$ is the absorbing element for $\bullet$.
Hence,    the family $C_{\mathrm{null}}$ defined by
$C_{\mathrm{null}}(\cdot, 0) = E$ and  
$C_{\mathrm{null}}(\cdot, y) = Z$
for all $y > 0$ is yet another example of a subsemigroup of $\mathcal C$. 
 
The following theorem  says that the above examples exhaust the possibilities.

 \begin{theorem}\label{th:classify}
Suppose 
$$
C(\kappa,0) = (1-\kappa)^+ \mbox{ for all } \kappa \ge 0
$$
and 
$$
C(\cdot, y_1) \bullet C(\cdot, y_2) = C(\cdot, y_1+ y_2) \mbox{ for all } y_1, y_2 \ge 0.
$$
Then exactly one of the following holds true:
\begin{enumerate}
\item $C(\kappa, y) = (1-\kappa)^+$ for all $\kappa \ge 0, y > 0$;
\item $C(\kappa,y) = 1$ for all $\kappa \ge 0, y > 0$;
\item $C=C_f$ for a 
 log-concave density $f$.
\end{enumerate}
In case (3) the density $f$ is uniquely defined by the semigroup, 
up to centring and scaling.
\end{theorem}

The proof appears in Section \ref{se:proofs}.

\begin{remark}
One could certainly consider other binary operations  on the space $\mathcal C$
which are also compatible with the partial order.  
For instance, we could let
\begin{align*} 
C_1 \clubsuit C_2(\kappa) = 1 - \EE[S_1 \wedge (\kappa S_2^*)  ] 
\end{align*}
where the primal representation $S_1$ of $C_1$
 is   \textit{independent} of the dual representation $S_2^*$ of $C_2$. 
 Note that this binary operation $\clubsuit$ is commutative, 
 and indeed we have
 $$
C_1 \clubsuit C_2(\kappa) = 1 - \EE[ (S_1 S_2) \wedge \kappa]
$$
where $S_2$ is a primal representation of $C_2$, again independent 
of $S_1$.   In fact, the binary operation $\clubsuit$ 
can be expressed (somewhat awkwardly) in terms of the call price curves $C_1$ and $C_2$:
\begin{align*} 
C_1 \clubsuit C_2(\kappa) = 1 + \int_0^{\infty} \frac{\kappa}{\eta}
C_1'(\eta) C'_2(\kappa/\eta) d\eta - \int_0^{\infty}
 \int_0^{\infty} C_1'(\eta_1) C_2'(\eta_2) \one_{\{ \eta_1 \eta_2 \le \kappa\}} 
 d\eta_1 \ d\eta_2.
 \end{align*}

As described above, one could construct call price surfaces by
studying one parameter semigroups   
for this binary operation $\clubsuit$.  
Indeed, such semigroups are easy to describe since their primal
representations are essentially exponential L\'evy processes.   Unfortunately,
the call prices given by an exponential
 L\'evy process are  not easy to 
write down in general.   However, we have seen that the 
one-parameter semigroup of call prices for the  binary operation $\bullet$
are extremely simple to write down. 
 It is the simplicity of these formulae that is the claim to practicality
 of the results presented here. 
\end{remark}

\section{Calibrating the surface}\label{se:calibrating}
\subsection{An exploration of $C_f$}\label{se:explore}
We have argued that if $f$ is a log-concave density and $\Upsilon$ is an increasing function, then the family
$\{ C_f(\kappa, \Upsilon(t)): \kappa \ge 0, t \ge 0 \}$  is a call
price surface as defined in Section \ref{se:surface}, where
the notation $C_f$ is defined in Section \ref{se:one-parameter}. 
The motivation of this section is to bring this observation
from theory to practice. In particular, to calibrate the functions $f$ and $\Upsilon$
to market data, it is useful to have at hand some properties,
including asymptotic properties, of the function $C_f$.  

In what follows we will assume that the density $f$ has support 
of the form $[L, R]$ for some constants 
$-\infty \le L < R \le + \infty$.  Recall that we assume $f$ is continuous
on $[L,R]$.
Now let $Z$ be random variable with density $f$. 
For each $y\in \RR$, define a non-negative random variable by
$$
S^{(y)} =  \frac{f(Z+y)}{f(Z)}.
$$
Note that $S^{(y)}$ is well-defined since $L < Z < R$ almost 
surely, and hence $f(Z) > 0$ almost surely.
In this notation, we have for all $y \in \RR$ that
$$
C_f(\kappa,y) = 1 - \EE[ S^{(y)} \wedge \kappa  ]
$$
so that by Theorem \ref{th:rep} we have $C_f(  \cdot, y) 
\in \mathcal C$  and that $S^{(y)}$ is a primal representation 
of $C_f(\cdot, y)$.

Note also that $S^{(y)} = 0$ almost surely for $|y| \ge R-L$
while for $|y| < R-L$ we have
$$
\PP( S^{(y)} > 0 ) =  \PP(L + y^- <Z < R - y^+)
$$ 
and
$$
\EE( S^{(y)} ) = \PP( L + y^+ < Z < R - y^-).
$$
In particular, for $y > 0$ we have
$$
C_f( \cdot, y) \in \mathcal C_+ \mbox{ if } R = + \infty
$$
and 
$$
C_f( \cdot, y) \in \mathcal C_1 \mbox{ if } L = - \infty,
$$
where the sets $\mathcal C_+$ and $\mathcal C_1$ are defined in Section \ref{se:binary}.

By changing variables, we find that 
a dual representation of $C_f(\cdot, y)$ is given by
$$
S^{(y)*} = \frac{f(Z-y)}{f(Z)} = S^{(-y)}
$$
 and therefore 
$$
C_f(\cdot,y)^* = C_f(\cdot,-y).
$$
It is interesting to observe that the call price surface 
$C_f$ satisfies the put-call symmetry formula $C_f(\cdot,y)^*  = 
C_f(\cdot,y)$
if the density $f$ is an even function.

 By implication (2) of Proposition~\ref{th:bobkov} 
we have that for $y \ge 0$ that the map $z \mapsto \frac{f(z
+y)}{f(z)}$ is non-increasing.  Let
$$
d_f(  \kappa,y) = \sup\left\{ z  > L:   \frac{f(z+y)}{f(z)} \ge 
\kappa \right\}  \mbox{ for } \kappa \ge 0, y \ge 0
$$  
with the convention that $\sup \emptyset = L$.
  Note 
 that$\frac{f(z+y)}{f(z)} \ge \kappa$
if and only if  $d_f(y, \kappa) \ge z$. 
With this notation, we have 
$$
C_f(\kappa, y) = F( d_f( \kappa,  y) + y) - \kappa F( d_f
(\kappa, y) ) \mbox{ for all } \kappa \ge  0, y \ge  0,
$$
where  $F(x) = \int_{L}^x f(z) dz$ is the cumulative 
distribution corresponding to $f$.

\begin{remark} The standard normal density $\varphi$ is log-
concave
and  we have the computation
$$
d_\varphi(\kappa,y) = - \frac{ \log \kappa }{y} - \frac{y}{2}
$$
yielding the usual Black--Scholes formula. 
\end{remark}

The first result may seem like a curiosity, but in fact is
a useful alternative formula for computing $C_f$ numerically, given
the density $f$.  In particular, the following formula does
not require the evaluation of the function $d_f$ defined above.
This is a generalisation of Theorem 3.1 of \cite{SIFIN}.
The proof is essentially the same, but included here for completeness.
We will use the notation
$$
\hat C_f(p,y) = F( F^{-1}(p) + y)
$$

\begin{theorem}\label{th:first-duality} For $\kappa, y \ge 0$, we have
$$
C_f(\kappa, y ) = \sup_{0 \le p \le 1} [\hat C_f(p, y) - p \kappa ]
$$
\end{theorem}

\begin{proof}
Fix $\kappa, p, y$ and let $z=F^{-1}(p)$. Note that
\begin{align*}
F(z+y) - \kappa F(z) &= \int_{-\infty}^{z} (f(u+y)-\kappa f(u)) du \\
&\le \int_{-\infty}^{z} (f(u+y)-\kappa f(u))^+ du \\
&\le  \int_{-\infty}^{\infty} (f(u+y)-\kappa f(u))^+ du \\
& = C_f(\kappa, y).
\end{align*}
Given $\kappa$, there is equality when $z= d_f(\kappa,y)$.
\end{proof}

The next result gives an asymptotic expression for call prices
at short maturities and close to the money.
In what follows, we will use the notation
\begin{align*}
H_f(x) &=   f(L) +  \int_{L }^{R } (f'(z) - f(z) x)^+ dz \\
& = f(R)-  \int_{L }^{R } f'(z) \wedge [ f(z) x ] dz.
\end{align*}
where $f'$ is the right-derivative of $f$.  Recall that $f'$ always
exists on the interval $(L,R)$.

\begin{theorem}\label{th:generator}  
As $\varepsilon \downarrow 0$ we have that
$$
\frac{1}{\varepsilon} C_f(e^{\varepsilon x}, \varepsilon) \to
H_f(x).
$$
\end{theorem} 
 
 \begin{proof}
 Let $a$ be a maximum of $f$ so that $f(z+\varepsilon) \ge f(z)$ 
for $  z \le a-\varepsilon$ and
$f(z+\varepsilon) \le f(z)$ for $  z \ge R-\varepsilon$.  
 We only consider the case $L < a < R$,
as the cases $a=L$ and $a = R$ are similar.

Fix $x$ and $a-L < \varepsilon < R-a$,  and write
$$
  \frac{1}{\varepsilon}C_f(e^{\varepsilon x}, \varepsilon) 
= I_1 + I_2 + I_3
$$
where

\begin{align*}
I_1 &= \frac{1}{\varepsilon}\int_a^{R-\varepsilon} (f(z+
\varepsilon) - e^{\varepsilon x} f(z) )^+ dz  \\
I_2 &= \frac{1}{\varepsilon}\int_{a-\varepsilon}^a (f(z+
\varepsilon) - e^{\varepsilon x} f(z) )^+ dz \\
I_3 &= \frac{1}{\varepsilon}\int_{L-\varepsilon}^{a-\varepsilon} 
(f(z+\varepsilon) - e^{\varepsilon x} f(z) )^+ dz.
\end{align*}
Note that for  $a \le z \le R - \varepsilon$ we have
$$
\frac{1}{\varepsilon}(f(z+\varepsilon) - e^{\varepsilon x} f(z) 
)^+
\le x^- \ f(z) 
$$
so 
$$
I_1 \to \int_a^R (f'(z) - x f(z))^+ dz
$$
by the dominated convergence theorem.  

For the second term, note that by the 
continuity of $f$ at the point $a$ we have
$$
\sup_{a-\varepsilon \le z \le a} |f(z+\varepsilon) - e^
{\varepsilon x} f(z) |
\to 0
$$
as $\varepsilon \downarrow 0$. In particular, we have
$
I_2 \to 0.
$

Finally, for the third term apply the put-call parity formula
to get
\begin{align*}
 I_3 &= \frac{1}{\varepsilon} \int_{L-\varepsilon}^{a-
\varepsilon} (f(z+\varepsilon) - e^{\varepsilon x} f(z) )  dz 
+  \frac{1}{\varepsilon} \int_{L-\varepsilon}^{a-\varepsilon}
(e^{\varepsilon x} f(z) - f(z+\varepsilon))^+   dz  \\
&\to f(a)  - x F(a) + \int_L^a ( x f(z) - f'(z) )^+ dz,
\end{align*}
again by the dominated convergence theorem.  The conclusion follows
from another application of put-call parity and recombining the integrals.
 \end{proof}

There are two interesting consequences of Theorem \ref{th:generator} above.
The first is that the density $f$ can be recovered from short time
asymptotics.  We will use the notation
$$
\hat H_f(p) = f \circ F^{-1}(p)
$$
for $0 \le p \le 1$.  The proof follows the same pattern
as that of Theorem \ref{th:first-duality}, so is omitted.

\begin{theorem}\label{th:dualofH}
For all $0 \le p \le 1$  we have
\begin{align*}
\hat H_f(p) =  \inf_{x \in \RR} [ H_f(x) + xp ].
\end{align*}
\end{theorem}

\begin{remark}  Given the function $\hat H_f$, we can recover $f$,
up to centring, as follows:
Fix $0 < p_0 < 1$ and set $F(0) = p_0$.   Then we have
$$
 F^{-1}(p)  = \int_{p_0}^p \frac{dq}{\hat H_f(q)}.
$$ 
\end{remark}

 We note 
in passing that the call price function $C_f$  satisfies  a
non-linear partial differential equation featuring the
function $\hat H_f$  when $f$ is suitably well-behaved enough:
\begin{proposition}\label{th:PDE}
Let $f$ be a strictly log-concave density supported on all of $
\RR$.  Suppose that $f$ 
is $C^1$ and
such that 
$$
\lim_{z \downarrow - \infty} \frac{f'(z)}{f(z)} = + \infty
\mbox{ and }  \lim_{z \uparrow + \infty} \frac{f'(z)}{f(z)} = - 
\infty.
$$
Then
\begin{align*}
\frac{\partial C_f}{\partial y}  = \kappa \ \hat H \left( - 
\frac{\partial C_f}{\partial \kappa} \right)   =  \hat H \left( 
C_f - \kappa \frac{\partial C_f}{\partial \kappa} \right)
\end{align*}
on $( \kappa, y) \in (0,\infty) \times (0, \infty)$.
\end{proposition}

\begin{proof} 
By log-concavity, we have for all $z \in \RR$ and $y > 0$ that
$$
\frac{f'(z+y)}{f(z+y)} \le \frac{1}{y} \log \frac{f(z+y)}{f(z)}
\le  \frac{f'(z)}{f(z)}
$$
and hence 
$$
\lim_{z \downarrow -\infty} \frac{f(z+y)}{f(z)} = +\infty \mbox{ 
and } 
\lim_{z \uparrow +\infty} \frac{f(z+y)}{f(z)} = 0.
$$
Also, by the strict log-concavity of $f$ the map 
$z \mapsto \frac{f(z+y)}{f(z)}$ is strictly decreasing.  This 
shows that
  $d_f(\kappa, y)$ is finite for all $\kappa  > 0$ and
that 
$$
f(d_f(\kappa, y) +y )= \kappa f( d_f(\kappa, y)).
$$
By the differentiability of $f$ and the implicit function 
theorem,
the function $d_f$ is differentiable on $(\kappa, y) \in (0,
\infty) \times (0, \infty)$.

One checks that
$$
\frac{\partial C_f}{\partial y} = f( d_f(\kappa, y) + y ).
$$
and 
$$
 \frac{\partial C_f}{\partial \kappa} =  - F( d_f(\kappa,y) ).
$$
The conclusion follows. 
\end{proof}

We now comment on a second interesting consequence of Theorem \ref{th:generator}. 
Note that for the limit for the Black--Scholes call function is  
$$
\frac{1}{\varepsilon} C_{\mathrm{BS}}(e^{\varepsilon x}, 
\varepsilon)
\to H_{\varphi}(x) = \varphi(x) - x \Phi(-x).
$$
The function $H_{\varphi}$ has an interesting financial 
interpretation.
Recall that in the Bachelier model, assuming zero interest 
rates,
the initial price of a call option of maturity $T$ and strike 
$K$ is given by
$$
C_{0,T, K} = \sigma \sqrt{T} H_{\varphi} \left( \frac{ K - S_0}
{\sigma \sqrt{T}} \right)
$$
where $S_0$ is the initial price of the stock, and $\sigma$ is 
its
arithmetic volatility.  Hence $H_{\varphi}$   can be
interpreted as a normalised call price function in the Bachelier 
model. 

Following the motivation of this section, we are interested not only
in the call price surface itself, but also in the corresponding implied
volatility surface.  Recall that the function $\IBS$ is defined by
$$
y = \IBS(\kappa, c) \Leftrightarrow \BS(\kappa,y) = c.
$$
We will use the notation
$$
Y_f(\kappa, y) = \IBS( \kappa, C_f(\kappa, y) ).
$$
Recall that if the normalised price of a call of 
moneyness $\kappa$ and maturity $t$ is given by
$C_f(\kappa, \Upsilon(t) )$, then the option's implied volatility is given by
$\frac{1}{\sqrt{t}} Y_f(\kappa, \Upsilon(t))$.

\begin{remark}  A word of warning:  We have noted that 
the function $\BS$ 
is the restriction of the function $C_{\varphi}$
to $[0,\infty) \times [0,\infty)$.  However,  it is \textit{not} the case
 that the function  $\IBS$ is a restriction of the function 
 $Y_{\varphi}$.  Indeed, the second argument of $\IBS$
 is a number $c$ in $[0,1]$ while the second argument of $Y_{\varphi}$
 is a number $y$ in $[0,\infty)$.
\end{remark}

With this set-up, we now present a result that gives the 
asymptotics of the implied volatility surface in the short maturity, close to the money limit. 
  
   \begin{theorem}
 We have as $\epsilon \downarrow 0$ that
$$
\frac{1}{\varepsilon} Y_f(1, \varepsilon) \to
 \sqrt{2\pi} \max_z f(z),
$$
and more generally, that 
$$
\frac{1}{\varepsilon} Y_f(e^{\varepsilon x}, \varepsilon) \to
\Lambda_{\mathrm{Ba}}(x, H_f(x) )
$$
where $\Lambda_{\mathrm{Ba}}(x, c )$ is defined by
$$
\Lambda_{\mathrm{Ba}}(x, c ) = \lambda \Leftrightarrow
\lambda H_{\varphi}(x/\lambda) = c.
$$ 
   \end{theorem}

\begin{proof}
Fix $x \in \RR$ and $\delta > 0$.   Let $\lambda = \Lambda_{\mathrm
{Ba}}(x, 
H_f(x) +   \delta )$.
By Theorem \ref{th:generator}, there exists $\varepsilon_0 > 0$ 
such
that
$$
\frac{1}{\varepsilon} C_f(e^{\varepsilon x}, \varepsilon)  \le 
H_f(x)
+ \tfrac{1}{2}\delta,
$$
while
$$
\frac{1}{\varepsilon} C_{\varphi}(e^{\varepsilon x}, \lambda 
\varepsilon)  
\ge  \lambda H_{\varphi}(x/\lambda) - \tfrac{1}{2} \delta
$$
for all $0 < \varepsilon < \varepsilon_0$.   
Hence $C_f(e^{\varepsilon x}, \varepsilon) \le \BS(e^
{\varepsilon x}, \lambda \varepsilon)$ and hence
$$
\limsup_{\varepsilon \downarrow 0} \frac{1}{\varepsilon}
Y_f(e^{\varepsilon x}, \varepsilon) \le 
\liminf_{\delta \downarrow 0} \Lambda_{\mathrm{Ba}}(x, H_f(x) + \delta 
).
$$
A lower bound is established similarly.  The conclusion follows 
form the continuity of $\Lambda_{\mathrm{Ba}}$.
 \end{proof}

For the final theorem of this section, we fix the maturity date and 
now compute extreme strike
asymptotics of the implied volatility.  In what follows, 
we will say that an eventually positive 
function $g$  varies regularly at infinity  with exponent $
\alpha$ iff
$$
\frac{g(\lambda x)}{g(x)} \to \lambda^\alpha \mbox{ as } x \to   
\infty
\mbox{ for all } \lambda > 0.
$$  Regular variation at zero is defined similarly.
   
\begin{theorem}\label{th:wings}
Suppose that $f$ is a log-concave density such that
$-\log \circ f \circ \log$ varies regularly at infinity with 
exponent $a > 0$
and varies regularly at zero with exponent $ - b < 0$.
Then for $y >0 $ we have
$$
\limsup_{\kappa \uparrow \infty} \frac{ Y_f(\kappa,y)}{\sqrt
{\log \kappa}}
 = \sqrt{ 2 \tanh \left( \frac{by}{4} \right)}
 $$
 and
 $$
\limsup_{\kappa \downarrow 0} \frac{ Y_f(\kappa,y)}{\sqrt{-\log 
\kappa}}
 = \sqrt{ 2 \tanh \left( \frac{ay}{4} \right)}
 $$
\end{theorem}

\begin{proof}
The key observation is that $f^\theta$ is Lebesgue integrable if
and only if 
$\theta > 0$.   	Indeed,
since $f$ is integrable and log-concave, there exist constants 
$A, B$ with $B>0$
such that $f(z) \le e^{A - B|z|}$, and hence $f^\theta$ is bounded
from above by an integrable function for $\theta > 0$ and bounded
from below by a  non-integrable function for $\theta \le 0$.

Fix $y > 0$. The moment generating function of $\log S^{(y)}$ is
calculated as 
\begin{align*}
M(p) &= \EE[ (S^{(y)})^p ] \\
& =  I_1 + I_2
\end{align*}
where 
$$
I_1 =  \int_{-\infty}^0 f(z+y)^p f(z)^{1-p} dz
 \mbox{ and }
I_2 = \int_0^{\infty}  f(z+y)^p f(z)^{1-p} dz
$$

By assumption
$$
\frac{ \log f(z+y) }{\log f(z)} \to e^{-by} \mbox{ as } z \to - 
\infty.
$$
or equivalently, $f(z+y) = f(z)^{e^{-by} + \delta(z)}$ where
$\delta(z) \to 0$ as $z \to -\infty$, yielding the expression
$$
 f(z+y)^p f(z)^{1-p} =  f(z)^{ 1 - p(1 - e^{-by} - \delta(z) )}.
$$
The exponent of $f$ on the right-hand side 
is eventually positive, implying $I_1$ is finite, if
$$
p < \frac{1}{1- e^{-by}},
$$ 
 and the exponent is eventually negative, implying $I_1$ is 
infinite, if
$$
p > \frac{1}{1- e^{-by}}.
$$

On the other hand,  $f(z+y) = f(z)^{e^{ay} + \varepsilon(z)}$ 
where
$\varepsilon(z) \to 0$ as $z \to \infty$.  
Writing 
$$
 f(z+y)^p f(z)^{1-p} =  f(z)^{ 1 + p(e^{ay} -1 + \varepsilon(z)) 
 }
$$
we see that for any $p \ge 0$, the exponent of $f$ on right-hand
side is eventually positive, implying $I_2$ is finite.

Therefore, we have shown that
$$
p^* = \sup\{ p \ge 1: M(p) < \infty\} = \frac{1}{1- e^{-by}}.
$$
By Lee's moment formula \cite{lee}, we have
$$
\limsup_{\kappa \uparrow \infty} \frac{ Y_f(\kappa,y)}{\sqrt
{\log \kappa}}
= \sqrt{2} \left(\sqrt{p^*} - \sqrt{p^*-1} \right)
$$
from which the first conclusion follows. The calculation of the 
left-hand wing is similar.
\end{proof}

 \subsection{A parametric example}\label{se:examples}
In this section we consider a parametrised family of 
log-concave densities in which several interesting calculations
can be performed explicitly.  We then try to fit this
family to real call price data as a proof-of-concept.

Consider family of densities of the form
$$
f  (x) = \frac{1}{Z  } \left\{
\begin{array}{ll} 
 e^{r(c+a)x -  r e^{ax} } &  \mbox{ if } x \ge 0 \\   
   e^{r(c-b)x -  r e^{-bx} } &  \mbox{ if } x < 0 
\end{array} \right.
$$		
for parameters $a,b,r > 0$ and real $c$, with normalising constant 
$$
Z   =  \frac{1}{a}r^{-r(1 + c/a)} \bar\Gamma(r; r(1 + c/a))  + 
 \frac{1}{b }r^{-r(1 - c/b)}\bar\Gamma(r; r(1 - c/b))   
$$
where $\bar \Gamma(x, \theta) = \int_x^{\infty} z^{\theta-1} e^{-z}dz$
is the complementary incomplete gamma function.
It is straightforward to check that $f$ is a log-concave 
probability density.

Letting $a=b=r^{-1/2}$ and $c=0$, and then sending $r \to \infty$
recovers the Black--Scholes model $f  \to \varphi$. Roughly speaking, 
$a$ controls the left wing, $b$ the right wing, $c$ the at-the-money skew, 
$r$ the at-the-money convexity.  Although there 
are four parameters, recall 
from Section \ref{se:one-parameter} that 
we have
$$
C_{f_{(\lambda a, \lambda b, \lambda c,r)}}(\kappa,  y ) = 
C_{f_{(a,b,c,r)}}(\kappa, \lambda y ) 
$$
for $\kappa \ge 0$, $y \ge 0$ and $\lambda > 0$.  Hence, there is no 
loss of generality if we insist, for instance, that $abr = 1$, leaving us with only three free parameters.

The distribution function is given explicitly by
$$
F(x)  = \left\{
\begin{array}{ll}  1 - \frac{1}{Za}   r^{-r(1 + c/a)}\bar\Gamma(re^{ax}; r(1 + c/a)) 
 & \mbox{ if } x \ge 0  \\
\frac{1}{Zb }r^{-r(1 - c/b)}\bar\Gamma(re^{-bx}; r(1 - c/b))   
&  \mbox{ if } x < 0 .
\end{array} \right.
$$
The function $d_f$ can be calculated explicitly
 when the absolute log-moneyness $|\log \kappa|$ is sufficiently large:
$$
d_f(\kappa, y) = \left\{
\begin{array}{ll}  \frac{1}{a} \log \left( \frac{  (c+a)y - \frac{1}{r}\log \kappa}{ e^{ay} -1  } \right) 
& \mbox{ for } \kappa \le e^{r(c+a)y - r(e^{ay} -1 )} \\
-\frac{1}{b} \log \left( \frac{ - (c-b)y +  \frac{1}{r}\log \kappa}{ 1 -e^{-by}  } \right) 
& \mbox{ for } \kappa \ge e^{ r(c-b)y + r(e^{by} -1 )}.
\end{array} \right.
$$
Otherwise  $d_f(\kappa, y)$ is the unique root $-y < d < 0$ of the equation
$$
(c+a) y + (a+b) d = e^{a d + ay} - e^{-b d} + \tfrac{1}{r} \log \kappa,
$$
which can be calculated numerically, for instance,
by the bisection method.

The call price curve can be calculated 
by the formula
$$
C_f(\kappa, y) =  F( d_f(\kappa, y) + y ) - \kappa F(d_f(\kappa, y) ).
$$
Note that this formula is  rather explicit when
the absolute log-moneyness  is sufficiently large, and
furthermore, it is numerically tractable in
all cases.

This choice of $f$ has the advantage that call prices can be calculated
very quickly.   Also, for other vanilla options, 
 numerical integration is very efficient since the density function $f$ is
 smooth and decays quickly at infinity.  Alternatively,
 rejection sampling is available, since the density is bounded, for instance,
 by a Gaussian density. 
 
When it comes to calibrate the model, we must find parameters
$a,b,c,r$ and an increasing function $\Upsilon$ such that
$$
C_{f_{(a,b,c,r)}}(\kappa,  \Upsilon(t) )  \approx C^{\mathrm{obs}}(\kappa, t) \mbox{ for all } (  \kappa, t) \in \mathcal S
$$
where $C^{\mathrm{obs}}(\kappa, t)$ is the observed normalised 
price of a call option of moneyness $\kappa$ and maturity $t$,
where $\mathcal S$ is the set of pairs $(\kappa, t)$ for which
there is available market data.
Equivalently, we fit the parameters  
$a,b,c,r$ and the function $\Upsilon$  to try to approximate the
observed implied volatility surface.  

For this exercise, I downloaded E-mini S\&P MidCap 400 call options 
call and put option price data from \texttt{ftp://ftp.cmegroup.com/pub/settle/stleqt} on 12 July 2018, for 
maturities $t_1=0.2, t_2= 0.4, t_3=0.7$ years for all available strikes.
 Letting
$\mathcal S_i = \{ \kappa: (\kappa, t_i) \in \mathcal S \}$ be the set of available strikes for maturity $t_i$, we have
$|\mathcal S_1| = 251,  |\mathcal S_2| = 248$ and $|\mathcal S_3| = 232$
observations. There are six parameters to 
 find: $a, c,r $ and $\Upsilon(t_1) = y_1$, $\Upsilon(t_2) = y_2, \Upsilon(t_3) = y_3$ to fit  $251+248+232=731$ observations.  

To speed up the calibration, we can use the asymptotic 
implied total standard deviation
calculations of Section \ref{se:explore}.  In particular,
we can apply  Theorem \ref{th:wings}  by 
noting that 
$$
-\log \circ f \circ \log(x) = \log Z + \left\{
\begin{array}{ll} r x^{a} - r(a+c) \log x  & \mbox{ if } x \ge 1 \\
r x^{-b} + r(b-c) \log x  & \mbox{ if } x < 1.
\end{array}  \right.
$$
However, we can do better and replace each limsup with
a proper limit by applying standard asymptotic properties of the complementary
incomplete gamma function and the tail-wing formula of Benaim--Friz 
\cite{BenaimFriz} and   Gulisashvili \cite{gulisashvili} 
to find  
$$
\frac{Y_f(\kappa, y) }{\sqrt{\log \kappa}}
\to   \sqrt{2  \tanh\left( \frac{ b y}{4} \right)} \mbox{ as } \kappa \to \infty
$$
and
$$
\frac{Y_f(\kappa, y) }{\sqrt{-\log \kappa}}
\to\sqrt{2  \tanh\left( \frac{  a y}{4} \right)}\mbox{ as } \kappa \to 0.
$$

 Figure  \ref{fi:calibrated} shows  a calibration
 of this family of call prices to real market data.  
It is 
important to stress that there is no a priori reason why this
data should resemble the call surfaces generated by this family of models.
Nevertheless,  although the fit is 
 not perfect, it does seem to indicate that this modelling approach
 is worth pursuing further.

\begin{figure} 
\caption{Implied volatility vs. log-moneyness for market data versus fitted density (red) with $a=3.63, b =0.0545, c = -0.0665, r=  6.89$ and 
$y_1 = 0.234, y_2 = 0.356, y_3 = 0.439$.
}
	  \label{fi:calibrated}
		\includegraphics[trim = 0cm 0 0cm 0, clip, scale = 0.17]{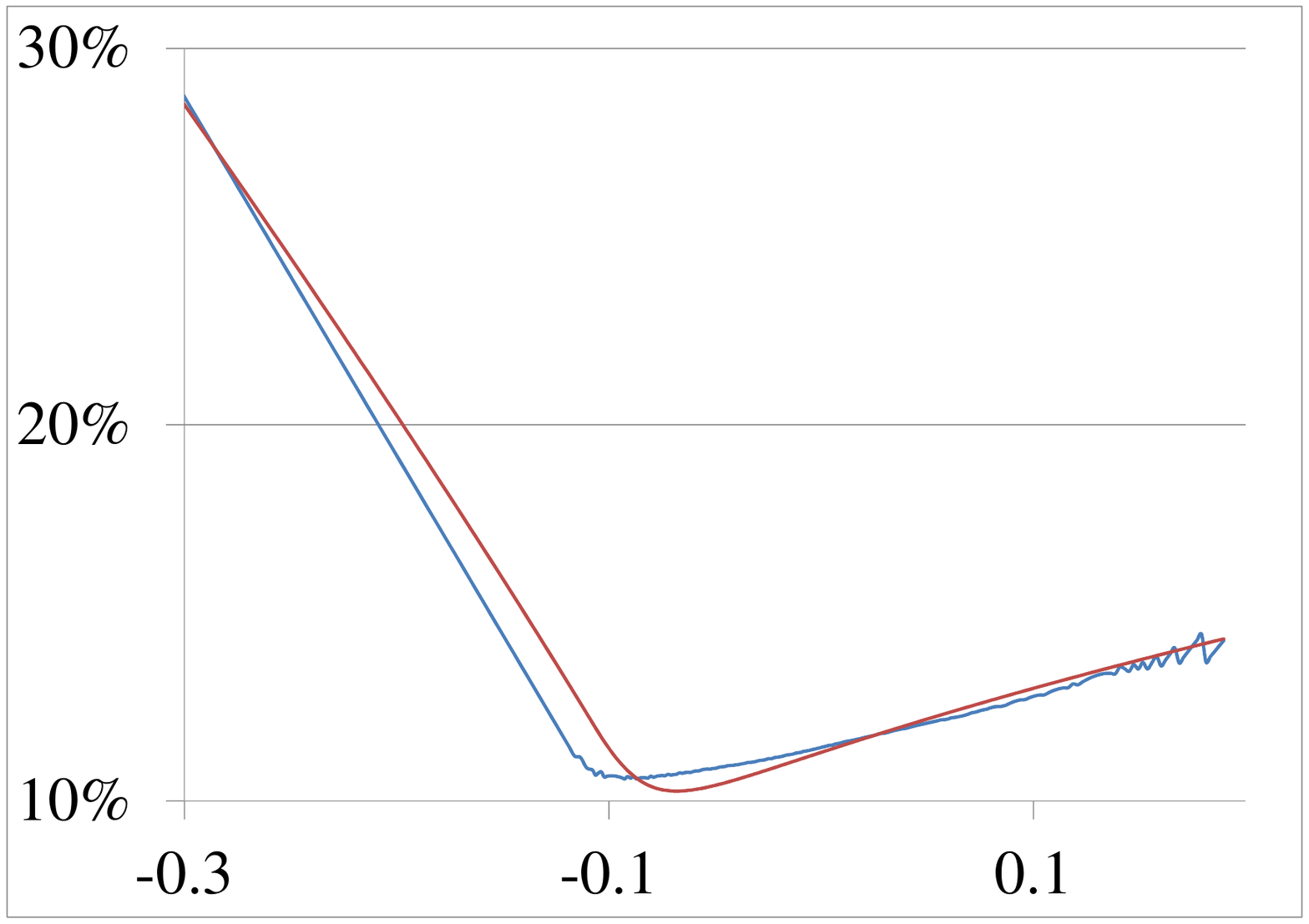}	
		\includegraphics[trim = 0cm 0 0cm 0, clip, scale = 0.17]{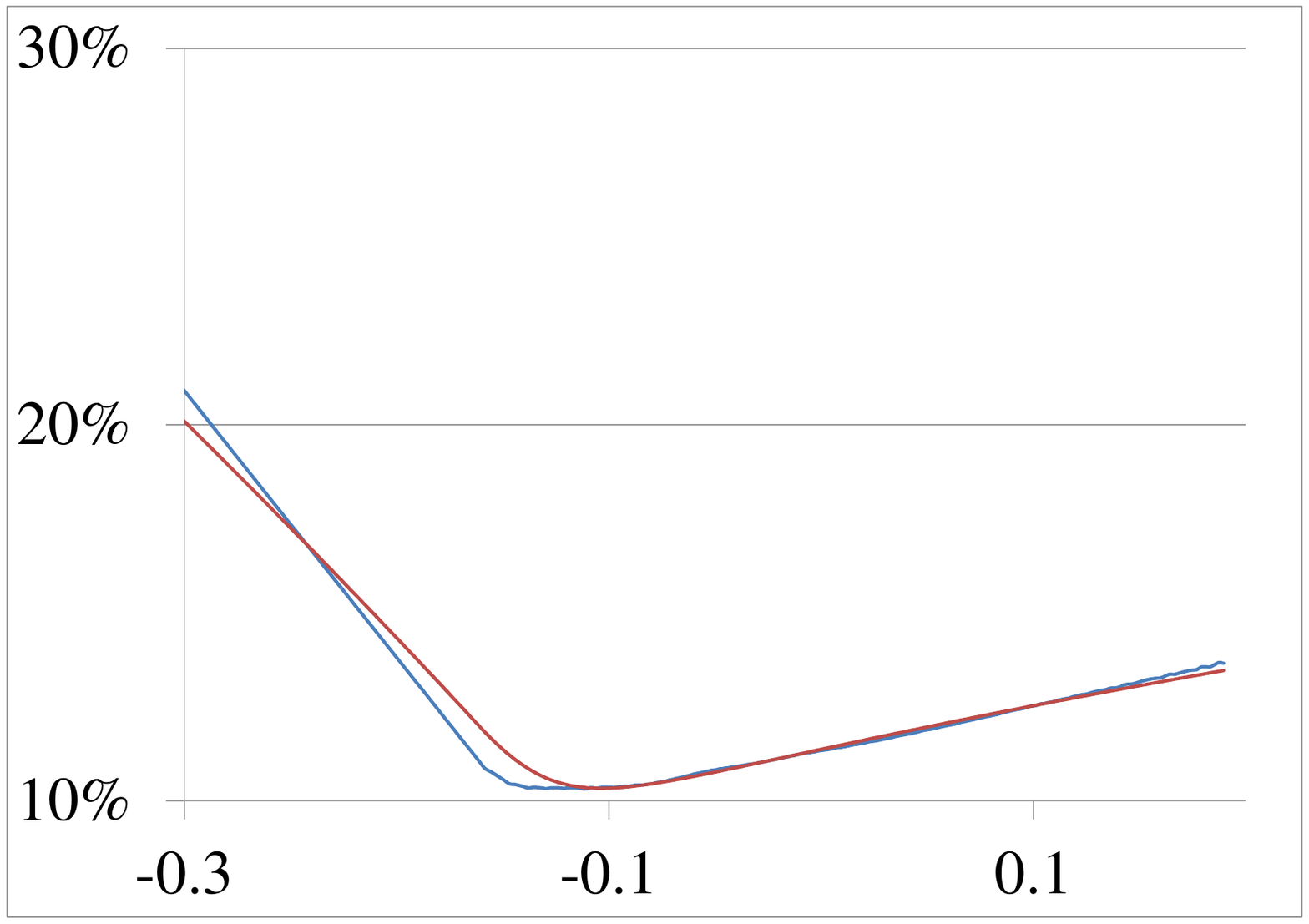}	
		\includegraphics[trim = 0cm 0 0cm 0, clip, scale = 0.17]{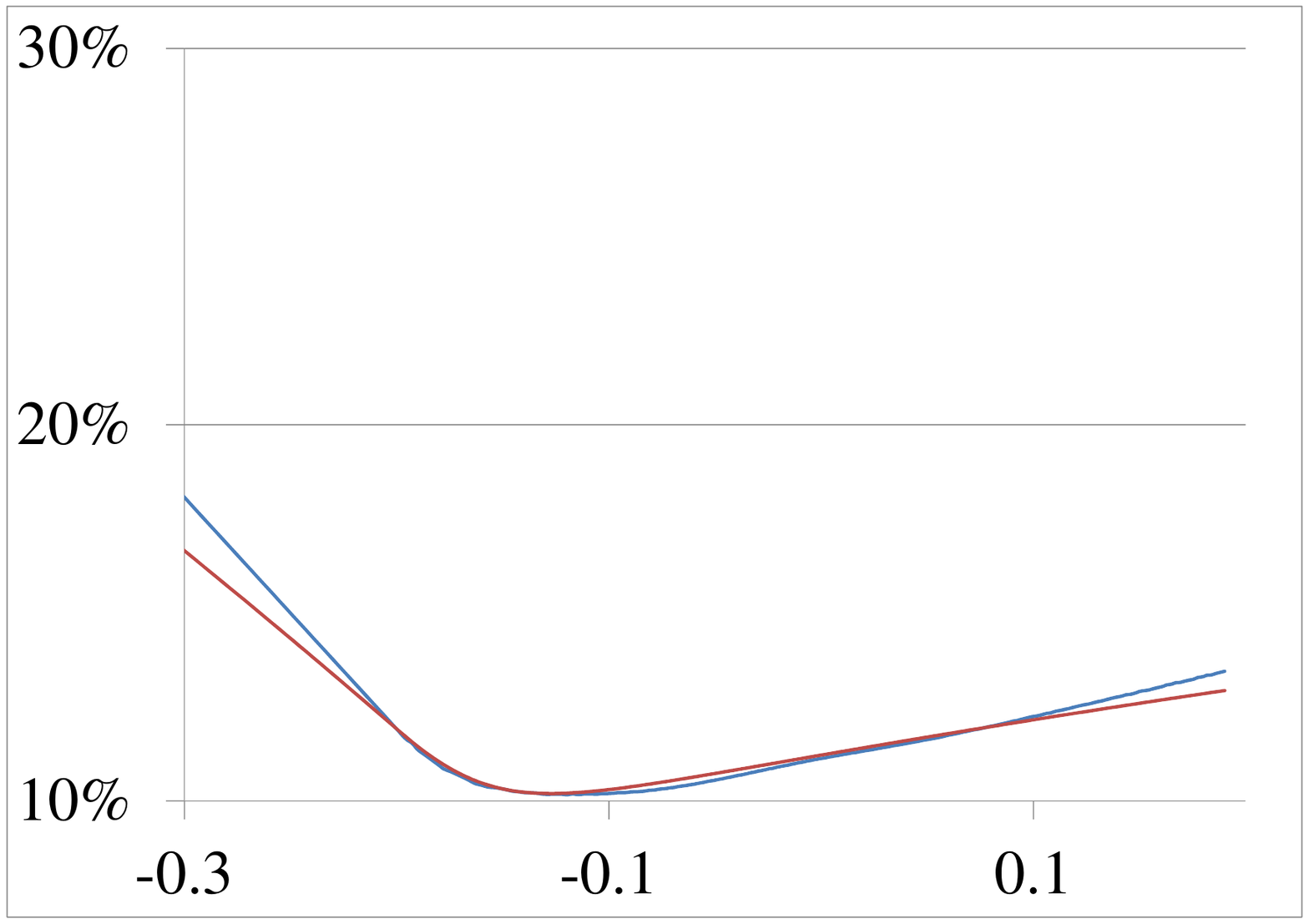}	
	\end{figure} 
	
	\subsection{A non-parametric calibration}
	In this section, we take a somewhat
	different approach.  Rather than assuming that the log-concave
density $f$ is a fixed parametric family, we use the 
results of section \ref{se:explore} to estimate $f$ non-parametrically. In particular, 
we assume that 
$$
C^{\mathrm{obs}}(\kappa, t_1) \approx C_{f }(\kappa,  \Upsilon(t_1) )  
 \mbox{ for all } \kappa \in \mathcal S_1,
$$	
where now the function $f$ is unknown.  Since the fit of the parametric
model was reasonably good, we will set $\Upsilon(t_1)$ to be
the same value $y_1$ found in section \ref{se:examples}.

Recall that Theorem \ref{th:generator}
says that
$$
C_f(e^{\varepsilon x}, \varepsilon) = \varepsilon H_f(x) + o(\varepsilon).
$$
It is straightforward to check that if $f$ satisfies as mild regularity condition
as in the hypothesis of Proposition \ref{th:PDE} then we have the
slightly improved asymptotic formula
$$
 C_f(e^{\varepsilon x}, \varepsilon)e^{-\varepsilon x/2}  = \varepsilon H_f(x) + o(\varepsilon^2).
$$

Hence, we will assume that 
$$
C_{f }(\kappa,  y_1 ) \kappa^{-1/2} \approx y_1 H_f( \log \kappa/y_1) 
$$
since $y_1$ is small.  Theorem \ref{th:dualofH} tells us that 
$$
f \circ F^{-1}(p) \approx \frac{1}{y_1}
\inf_{\kappa \in \mathcal{S}_1}  [ C^{\mathrm{obs}}(\kappa, t_1) \kappa^{-1/2}
+ p \log(\kappa)].
$$
An estimate of the density $f$ can now be computed numerically.  

Figure \ref{fi:pdf} compares $\log f$, when estimated non-parametrically
versus the calibrated parametric example from the last section.
Considering the fact that the non-parametric density is estimated from the earliest
maturity date, while the parametric density is calibrated using all three
maturity dates, the agreement is uncanny.

\begin{figure} 
\caption{$\log f$ estimated non-parametrically (blue), versus the parametric
fit (red). Both are centred so that their maxima are at the origin.
}
	  \label{fi:pdf}
		\includegraphics[ scale = 0.3]{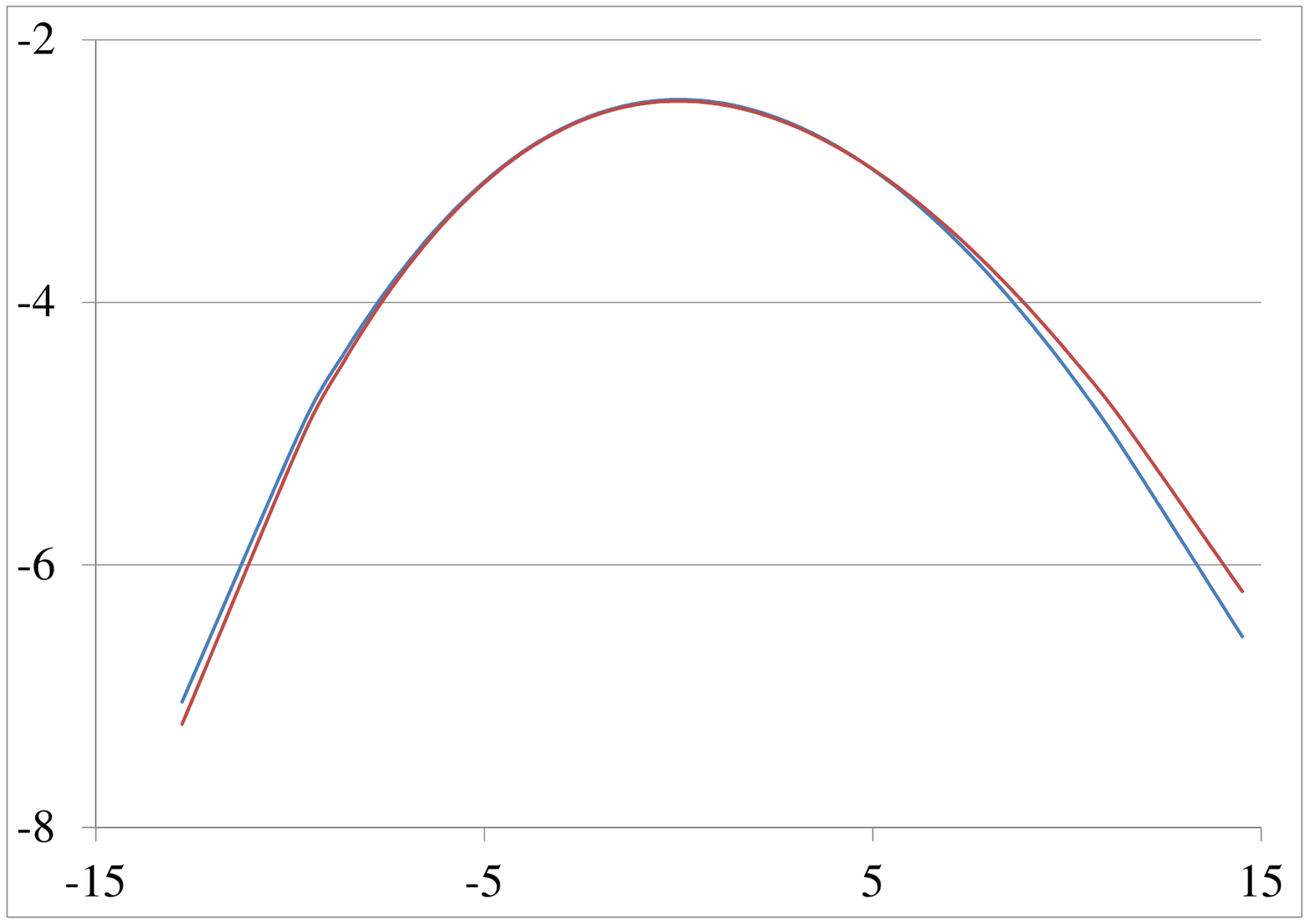}
	\end{figure} 

Given that the calibrated parametric density seems to recover market
date reasonably well, and that the non-parametric density agrees with
parametric reasonably well, it is natural to compare the market
implied volatility to that predicted by the non-parametric model.  Recall that 
the model 
call surface is determined by the density $f$ and the increasing 
function $\Upsilon$.	
We have estimated $f$ from the short maturity call prices
and the assumption that $\Upsilon(t_1) = y_1$, where $y_1$ was
found from the parametric calibration.  However, 
we still need to estimate the function $\Upsilon(t_i)$ for $i=2, 3$.
For a lack of a better idea, we let $\Upsilon(t_i) = y_i$ for $i=2,3$ as well.

Figure \ref{fi:nonparasmile} compares the market implied volatility
(the same as in figure \ref{fi:calibrated}), with the implied volatility
computed from the non-parametric model.  Since the estimated density $f$ is not
given by an explicit formula, I have used the formula in Theorem \ref{th:first-duality} to compute the call prices.
 Again, given that the  density $f$
is estimated using only the $t_1$ call price curve, it is interesting
that the model implied volatility for maturities $t_2$ and $t_3$ should match
the market data at all.

\begin{figure} 
\caption{Implied volatility vs. log-moneyness from market data (blue), 
 versus the non-parametrically estimated density (red) 
}
	  \label{fi:nonparasmile}
		\includegraphics[trim = 0cm 0 0cm 0, clip, scale = 0.17]{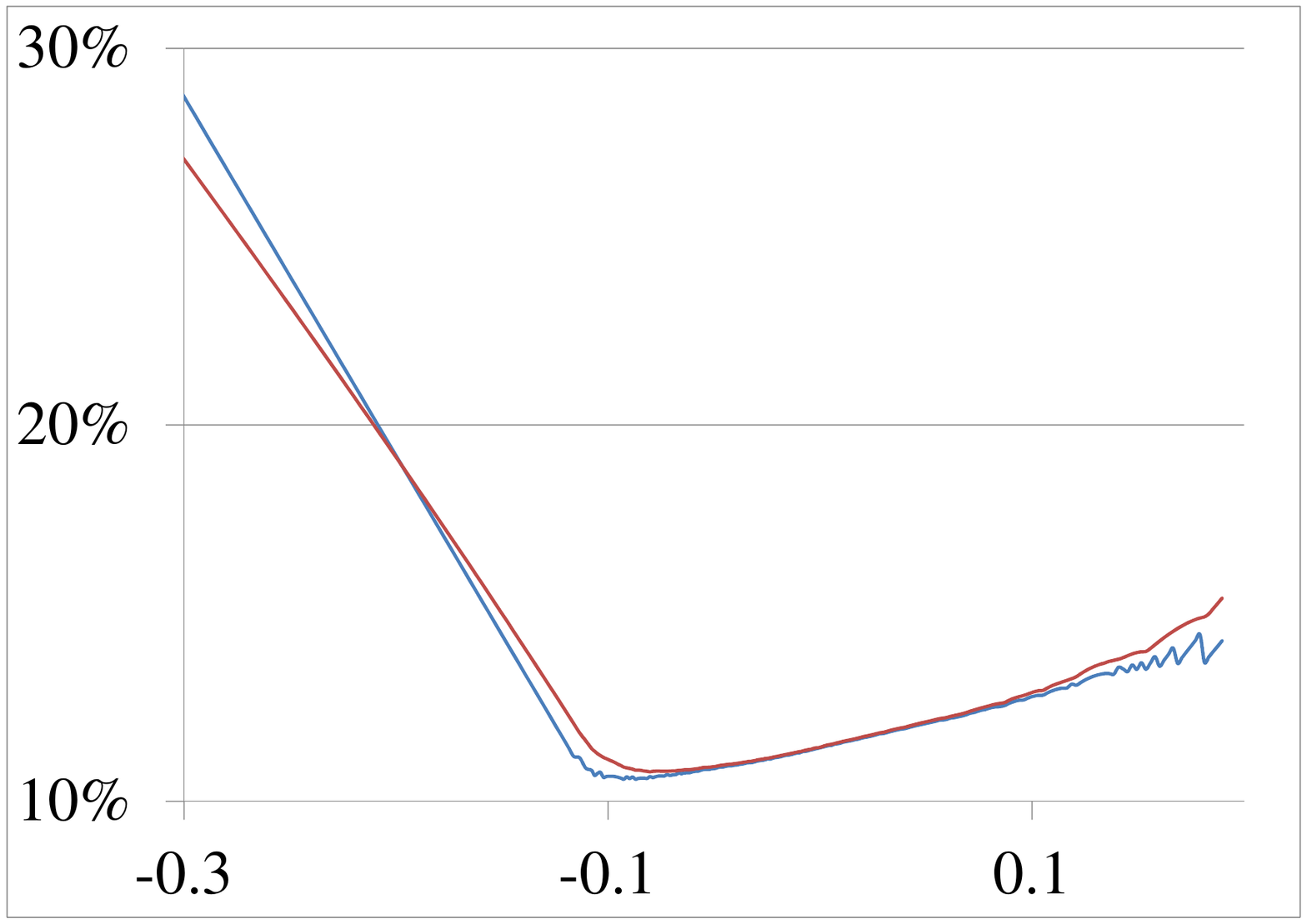}	
		\includegraphics[trim = 0cm 0 0cm 0, clip, scale = 0.17]{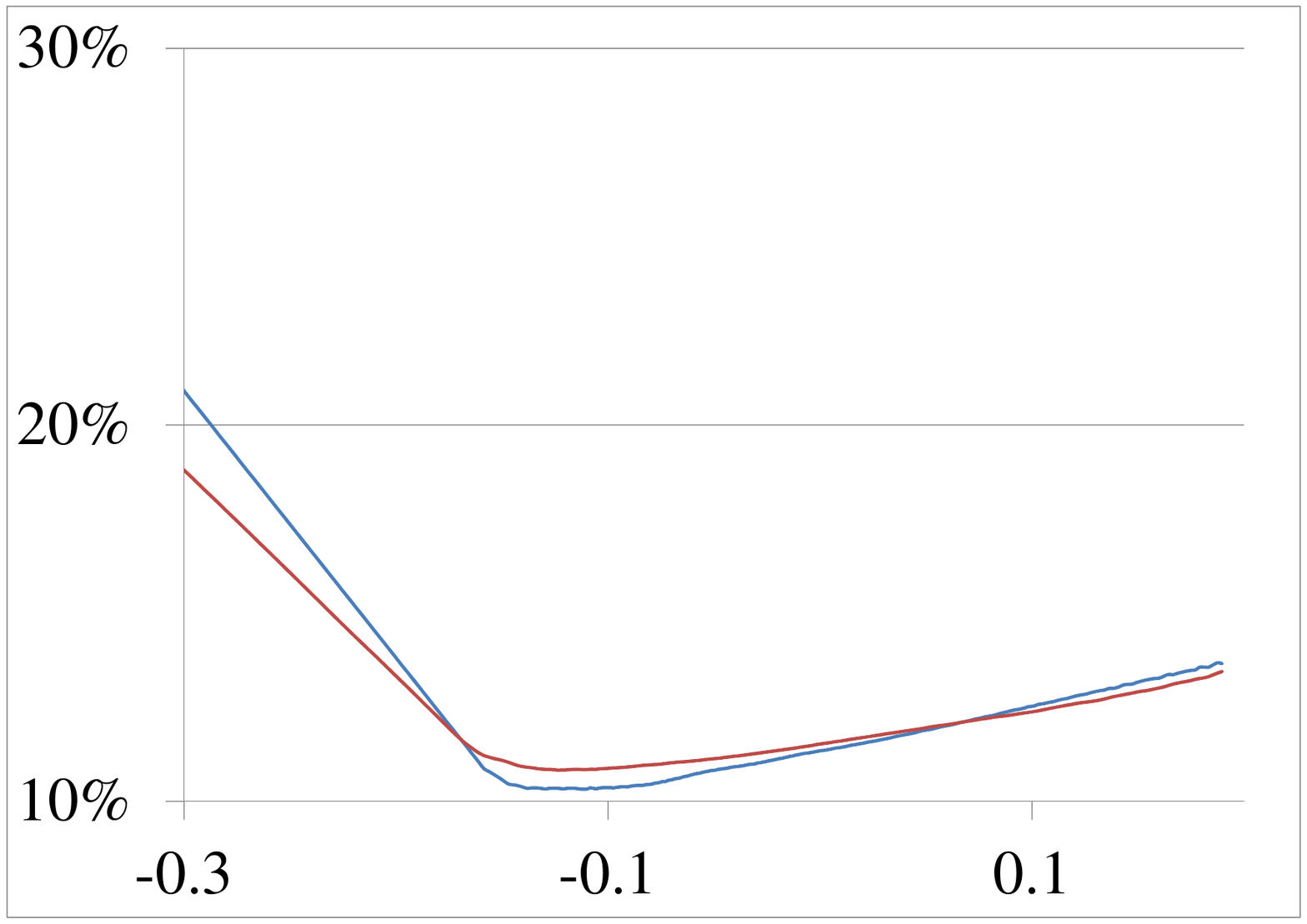}	
		\includegraphics[trim = 0cm 0 0cm 0, clip, scale = 0.17]{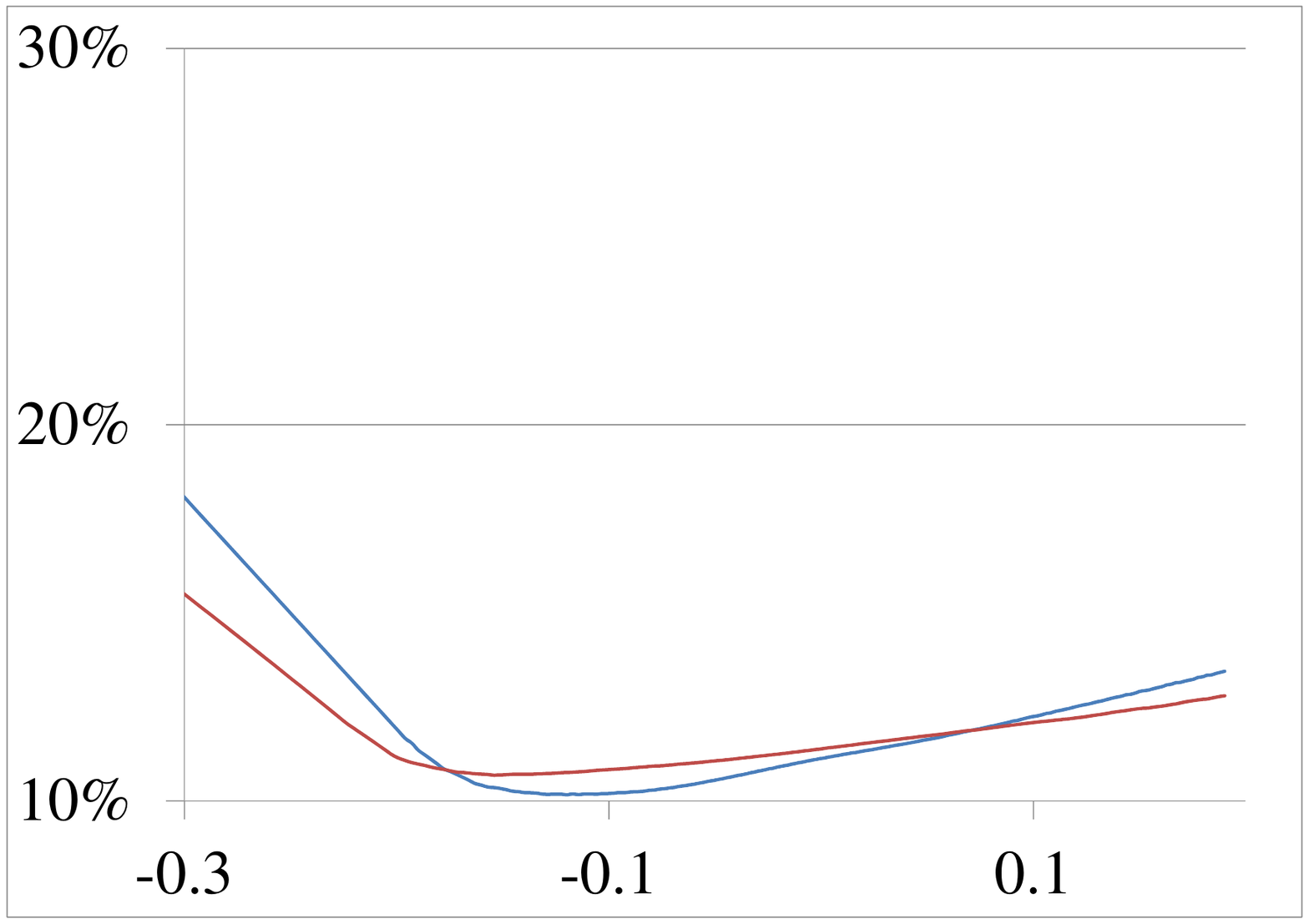}	
	\end{figure} 
	
\section{An isomorphism and lift zonoids}\label{se:proofs}
\subsection{The isomorophism}\label{se:iso}
In this section, to help understand the binary operation $\bullet$  on the space $\mathcal C$
we show that there is a nice isomorphism of $\mathcal C$ to 
another function space which 
converts the somewhat complicated operation $\bullet$ into simple function composition $\circ$.

We introduce a transformation $\hat{ \ }$ on the space $\mathcal C$ which will be particularly useful:
for $C \in \mathcal C$ we define a new function $\hat C$ on $[0,1]$ by the formula
$$
\hat C(p) = \inf_{\kappa \ge 0} [ C(\kappa) + p\kappa ] \mbox{ for } 0 \le p \le 1.
$$
We quickly note that 
the notation $\hat{ \ }$ introduced here is, in  fact, 
consistent with the prior occurrence of this notation in Section \ref{se:explore}.  Indeed, the connection between the  transformation $\hat{ \ }: \mathcal C \to \hat{ \mathcal C}$ defined here and 
the conclusion of Theorem \ref{th:dualofH} is explored in Section \ref{se:inf-conv} below.

\begin{figure}  	
\caption{A typical element of $\hat{\mathcal C}$}
	  \label{fi:hatC}
		\includegraphics[trim = 0cm 0 0cm 0, clip, scale = 0.35]{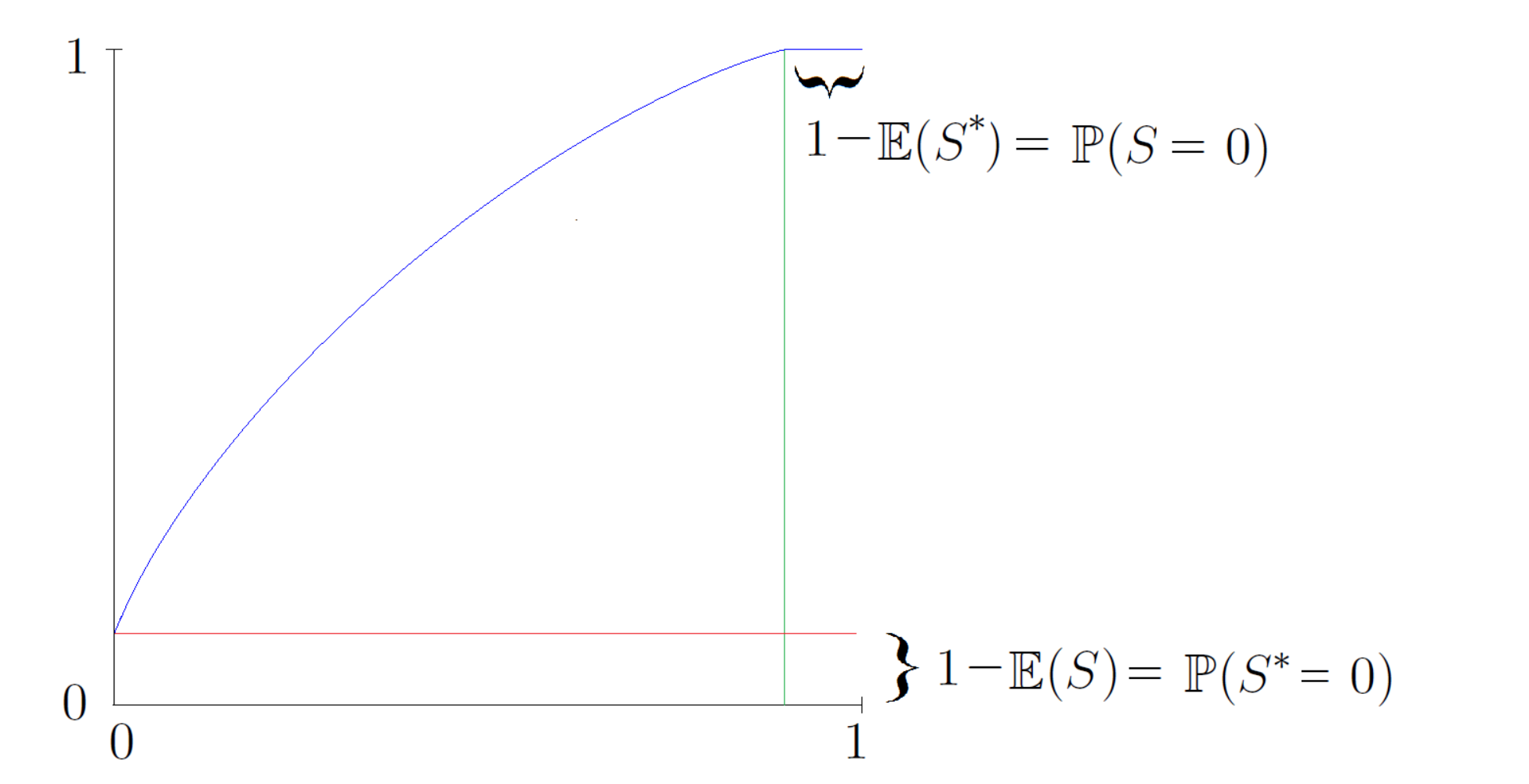}		
	\end{figure} 
	
	Given a call price curve $C \in \mathcal C$, we can immediately read off some properties of the new function $\hat C$.
The proof is routine, and hence omitted.
\begin{proposition}\label{th:propsofChat} Fix $C \in \mathcal C$ with primal representation $S$ and dual representation $S^*$.
\begin{enumerate}
\item $\hat C$ is non-decreasing and concave.
\item $\hat C$ is continuous and
$$
\hat C(0) =   C(\infty) = 1- \EE(S) = \PP(S^* = 0 ).
$$ 
\item For $0 \le p \le 1$ and $ \kappa \ge 0$ such that
$$
\PP(S > \kappa  ) \le p \le \PP(S \ge \kappa ),
$$ 
we have
$$
\hat C(p) =   C(\kappa) + p\kappa.
$$ 
\item
$
\min\{ p \ge 0: \hat C(p) = 1 \} =  - C'(0) = \PP(S > 0) = \EE(S^*).
$ 
\item $\hat C(p) \ge p$ for all $0 \le p \le 1$.
\end{enumerate}
\end{proposition}

Figure \ref{fi:hatC} plots the graph of a typical element $\hat C \in \hat{\mathcal C}$.

The next result identifies the image $\hat{\mathcal C}$ of the map $\hat{ }$,
and further shows that $\hat{ \ }: \mathcal C \to \hat{\mathcal C}$ is a bijection:
\begin{theorem}\label{th:duality}
Suppose $g:[0,1] \to [0,1]$ is continuous and concave with $g(1)=1$.  Let
$$
C(\kappa) =  \max_{0 \le p \le 1} [ g(p) - p \kappa] \mbox{ for all } \kappa \ge 0.
$$
Then $C \in \mathcal C$ and 
$
g = \hat{C}.
$ 
\end{theorem}

The above theorem is a minor variant of the Fenchel biconjugation theorem   of convex analysis. See the book of 
 Borwein \& Vanderwerff \cite[Theorem 2.4.4]{BV}.

The following theorem explains our interest
in the  bijection $\hat{ }$ : it converts the binary operation $\bullet$ to
function composition $\circ$.  
A version of this result can be found in the book of Borwein \& Vanderwerff
\cite[Exercise 2.4.31]{BV}.
\begin{theorem}\label{th:isomorphism}
For $C_1, C_2 \in \mathcal C$ we have 
$$
\widehat{C_1 \bullet C_2} = \hat{C_1} \circ \hat{C_2}
$$
\end{theorem}

\begin{proof}  By the continuity of a function $C \in \mathcal C$ at
$\kappa=0$,  we 
have the equivalent expression 
$$
\hat C(p) = \inf_{\kappa > 0} [ C(\kappa) + p\kappa ] \mbox{ for } 0 \le p \le 1.
$$
Hence for any $0 \le p \le 1$ we have
\begin{align*}
\widehat{C_1 \bullet C_2}(p) & =  \inf_{\kappa > 0} [ C_1 \bullet C_2 (\kappa) + p\kappa ] \\ 
& =  \inf_{\kappa > 0} \{ \inf_{H > 0} [ C_1(H) + H C_2 (\kappa/H) ] + p\kappa \} \\ 
& =  \inf_{H > 0} \{ C_1(H) +  H \inf_{\kappa > 0} [ C_2 (\kappa)  + p\kappa] \} \\ 
& = \hat{C_1} \circ \hat{C_2}(p).
\end{align*}
\end{proof}

In light of Theorem \ref{th:isomorphism}, Theorem \ref{th:semigroup}
says that the set of conjugate functions
$\hat {\mathcal C}$ is a semigroup with respect to function composition
$\circ$, with identity element $\hat{E}(p) = p$ and absorbing element
$\hat Z(p) = 1$.   The involution on $\hat{\mathcal C}$ induced by
${}^*$  is identified in Theorem \ref{th:inverse} below.

In preparation for reproving Theorem \ref{th:peacock}  and proving Theorem \ref{th:classify}
we identify the image of the set of functions $C_f$ under the isomorphism $\hat{ \ }$.  As the notation introduced in section \ref{se:explore} suggests,
we have
$$
\widehat{C_f(\cdot, y)}(p)   = F( F^{-1}(p) + y)  
$$
by Theorems \ref{th:first-duality} and \ref{th:duality}.  
For notational ease, we will continue to use the notation
$$
\hat C_f(p,y) = \widehat{C_f (\cdot, y)}(p).
$$

\begin{proof}[Another proof of Theorem \ref{th:peacock}]
Note that the family of functions $ ( \hat C_f(\cdot, y) )_{y \ge 0}$ form
a semigroup with respect to function composition.  The result
follows from applying Theorems \ref{th:duality} and \ref{th:isomorphism}.
\end{proof}

We now come to proof of Theorem \ref{th:classify}.

\begin{proof}[Proof of Theorem \ref{th:classify} ]
If a function $C:[0,\infty)\times [0,\infty) \to [0,1]$
satisfies the hypotheses of the theorem, then the conjugate function $\hat C:[0,1]\times [0,\infty) \to [0,1]$
is such that 
$$
\hat C(p,0) = p \mbox{ for all } 0\le p \le 1
$$
and satisfies  the translation equation 
$$
\hat C( \hat C(p,y_1), y_2) = \hat C(p, y_1+y_2) \mbox{ for all } 0 \le p \le 1 \mbox{ and } y_1, y_2 \ge 0.
$$
 The conclusion of the theorem is that there are only three types of solutions to the above
functional equation such that $\hat C( \cdot, y) \in \hat{ \mathcal C}$ for all $y > 0$:
\begin{enumerate}
\item $\hat C(p,y) = p$ for all $0\le p \le 1$ and $y > 0$,
\item  $\hat C(p, y) = 1$  for all $0\le p \le 1$ and $y > 0$,
\item $\hat C(p, y) = F( F^{-1}(p) + y )$ for all $0\le p \le 1$ and $y > 0$
where $F(z) = \int_{-\infty}^z f(x) dx$ and $f$ is a log-concave probability density.
\end{enumerate}

Once we have ruled out cases (1) and (2), we can appeal to 
the result of  Cherny \& Filipovi\'c \cite{CF}: 
 concave solutions of the translation equation on $[0,1]$ are of the form
$
 G^{-1}( G(\cdot) + y )
$
 where 
$$
G(p) = \int_{p_0}^p \frac{ dq}{ \hat H(q) }
$$
for a positive concave function $\hat H$ and fixed $0 < p_0 <1$.
Note that for $0 < p < 1$ the integral is well-defined and finite as $\hat H$ is positive and continuous by concavity.
Let $L = G(0)$ and $R= G(1)$, and define a function $F:[L,R] \to [0,1]$ as the inverse function
$F = G^{-1}$, and extend $F$ to all of $\RR$ by $F(x) = 0$ for $x \le L$ and $F(x) = 1$ for $x  \ge R$.
Note that we can compute the derivative as
$$
F'(x) = \frac{1}{G' \circ G^{-1}(x)} = \hat H ( F(x) ) \mbox{ for all } x \in \RR.
$$
Setting $f= F'$, we have $\hat H = f \circ F^{-1}$.   
Since $\hat H$ is concave,  Bobkov's result Proposition \ref{th:bobkov} 
implies that $f$ is log-concave.

\end{proof}

\begin{remark} An earlier study of the translation equation  without
the concavity assumption can
be found in the book of Acz\'el    \cite[Chapter 6.1]{aczel}.
\end{remark}

\subsection{Infinitesimal generators and the inf-convolution}\label{se:inf-conv} In this section we briefly and informally discuss the connection between
the binary operation $\bullet$ defined in section \ref{se:binary}
and the well-known inf-convolution $\square$.

Let $f$ be a log-concave density with distribution function $F$, and
let
$$
\hat C(p, y) = F( F^{-1}(p) + y) \mbox{ for all } 0 \le p \le 1, y \ge 0.
$$
The content of Theorem \ref{th:classify} is that, aside from the
trivial and null semigroups,  the only semigroups of
$\hat {\mathcal C}$ with respect to composition are of the above form.
The infinitesimal generator is given by
$$
\left. \frac{\partial}{\partial y} \hat{C}(p, y) \right|_{y=0} = \hat H(p)
 \mbox{ for all } 0 \le p \le 1,
$$
where $\hat H = f \circ F^{-1}$  and  we have taken the version of $f$ which is continuous on its support $[L,R]$.  Note that this equation also
holds for the trivial semigroup with $\hat H = 0$.

The key property of the function $\hat H$ is that it is non-negative
and concave.  
Let 
$$
\hat{ \mathcal H} = \{ h:[0,1] \to [0,\infty), \mbox{ concave  } \}.
$$
Note that for every element of $\hat{\mathcal{H}}$, aside from $\hat H=0$,  one can assign a  unique (up to centring) log-concave density $f$ by the discussion of section \ref{se:explore}.

The space $\hat{\mathcal H}$   is closed under addition.  Furthermore, we have
for every non-null one-parameter semigroup $\hat C$ that
$$
\hat C(p, \varepsilon) \approx p + \varepsilon \hat H(p) \mbox{ for small } \varepsilon > 0 
$$
for some $\hat H \in \hat{\mathcal{H}}$.  Let $\hat C_1$ and $\hat C_2$ be
two such semigroups.  Note that
\begin{align*}
\hat C_1(C_2(p, \varepsilon), \varepsilon) 
&    \approx p + \varepsilon (\hat H_1(p) + \hat H_2(p) )
\end{align*}
implying that
function composition near 
the identity element of $\hat{\mathcal C}$ amounts to addition
in the space of generators $\hat{ \mathcal H}$.

Similarly, let 
$$
\mathcal H = \left\{ H: \RR \to [0,\infty) \mbox{ convex with } 0 \le  H(x) - (-x)^+ \le \mbox{ const. }  \right\}.
$$
For $H \in \mathcal H$, let
$$
\hat H(p) = \inf_{x \in \RR} [H(x) + xp] \mbox{ for } 0 \le p \le 1.
$$
One can check that 
$\hat{ \ }$ is a bijection between the
sets $\mathcal H$ and $\hat{\mathcal H}$  
by a version of the Fenchel biconjugation theorem.  In particular, the space $\mathcal H$ 
can be identified with the generators of one-parameter
 semigroups in $ {\mathcal C}$.

Recall that the inf-convolution of two    functions $f_1, f_2: \RR \to \RR$ is defined by
$$
(f_1 \square f_2)(x) = \inf_{y \in \RR} [ f_1(x-y) + f_2(y) ] \mbox{ for } x \in \RR.
$$
The basic property 
of the inf-convolution (see \cite[Exercise 2.3.15]{BV} for example)  is that it becomes
addition under conjugation: 
\begin{align*}
\widehat{f_1 \square f_2}(p) &= \inf_{x \in \RR}   \inf_{y \in \RR} [ f_1(x-y) + f_2(y)    + xp ] \\
& = \inf_{z \in \RR} [ f_1(z) + zp ] +  \inf_{y \in \RR} [   f_2(y)    + yp ] \\
& = \hat f_1(p) + \hat f_2(p),
\end{align*}
in analogy with 
Theorem \ref{th:isomorphism}.  Since there
is an exponential map lifting function addition $+$  to 
function composition $\circ$ in  $\hat{\mathcal C}$,  we can apply the 
isomorphism $\hat{ \ }$ to conclude that there is an exponential
map lifting inf-convolution $\square$  to 
the binary operation $\bullet$ in $\mathcal C$.

Indeed, let $C$ be a one parameter semigroup with generator $H$, so that
$$
C(e^{\varepsilon x}, \varepsilon) \approx \varepsilon H(x) \mbox{ for small }
\varepsilon > 0.
$$ 
Letting $C_1$ and $C_2$ be two such semigroups, we have
\begin{align*}
C_1( \cdot, \varepsilon) \bullet C_2( \cdot, \varepsilon)(e^{\varepsilon x})
&\approx  \varepsilon  \ \inf_y [ H_1(y) + e^{\varepsilon y}H_2(x-y)]  \\
& \approx \varepsilon H_1 \square H_2(x)
\end{align*}

\subsection{Lift zonoids}
Finally, to see why one might want to compute  the Legendre transform
of a call price  with respect to the strike parameter, we 
recall that the zonoid  of an integrable random
$d$-vector $X$ is the set 
$$
Z_X = \left\{ \EE[X g(X)  ]  \mbox{ measurable } g:\RR^d \to [0,1]  \right\} \subseteq \RR^d,
$$
and that the lift zonoid of $X$ is the zonoid of the $(1+d)$-vector $(1,X)$ given
by
$$
\hat Z_X =  \left\{ ( \EE[g(X)], \EE[ X g(X)   ] )  \mbox{ measurable } g:\RR^d \to [0,1]  \right\} \subseteq \RR^{1+d}.
$$
The notion of lift zonoid was introduced in
the paper of Koshevoy \& Mosler \cite{KM}.

In the case $d=1$, the lift zonoid $\hat Z_{X}$ is a convex set 
contained in the rectangle
$$
[0,1] \times [-m_-, m_+ ].
$$
where $m_{\pm} = \EE( X^{\pm}).$  The precise shape of this set 
is intimately related to call and put prices as seen in  the following theorem.
\begin{theorem}\label{th:zonoid}
Let $X$ be an integrable random variable. Its lift zonoid is given by
$$
\hat Z_{X}= \left\{ (p, q):    \sup_{x \in \RR}\{ px - \EE[ (x-X)^+ ] \}
\le q \le \inf_{x \in \RR} \{ px + \EE[ (X-x)^+ ] \}, \ \ 0 \le p \le 1   \right\}.
$$
\end{theorem}

Note that if we let 
$$
\Theta(x) = \PP(X \ge x)
$$
then we have
$$
\EE[ (X-x)^+ ]  = \int_x^{\infty} \Theta(\xi) d\xi
$$
by Fubini's theorem.  Also 
if we define the inverse function $\Theta^{-1}$ for $0 < p < 1$ by
$$
\Theta^{-1}(p) = \inf\{ x: \Theta(x) \ge p \}
$$
then by a result of Koshevoy \& Mosler  \cite[Lemma 3.1]{KM} we have
$$
\hat Z_{X}= \left\{ (p, q):   \int_{1-p}^1 \Theta^{-1}(\phi) d\phi \le q \le 
 \int_{0}^p  \Theta^{-1}(\phi) d\phi, \ \ 0 \le p \le 1   \right\}.
$$
from which Theorem \ref{th:zonoid} can be proven by Young's inequality.
However since the result can be viewed as an 
application of the Neyman--Pearson lemma, we include a short proof for completeness. 
\begin{proof}
For any measurable function $g$ valued in $[0,1]$ and $x \in \RR$ 
we have 
$$
X g(X) \le  (X-x)^+ + x g(X)
$$
with equality when
$g$ is such that
$$
\one_{(x, \infty) } \le g  \le   \one_{[x, \infty)}. 
$$

Now suppose $(p,q) \in \hat Z_X$ so that $p = \EE[ g(X) ]$ and
$q = \EE[ X g(X) ]$ for some $g$.  Hence, computing 
expectations in the inequality above yields
$$
q \le \EE[ (X-x)^+ ] + xp.
$$
with equality if
$$
\PP(X > x ) \le p \le \PP(X \ge x ).
$$

By replacing $g$ with $1-g$, we see that $(p,q) \in \hat Z_X$ if and only if
$(1-p, \EE(X) - q ) \in \hat Z_X$, yielding the lower bound.
\end{proof}

We remark that the explicit connection between lift zonoids and the price of call options has
been noted before, for instance  in the paper of Mochanov \& Schmutz \cite{MS}.  
In the setting of this paper,  given $C \in \mathcal C$   represented by $S$,  
the lift zonoid of $S$ is given by the set
$$
\hat Z_S = \{ (p,q) :  1 - \hat C(1-p) \le q \le \EE(S) -1 + \hat C(p), \ \ 0 \le p \le 1 \}
$$ 

We recall that a random vector $X_1$ is dominated by $X_2$
in the lift zonoid order if $\hat Z_{X_1} \subseteq \hat Z_{X_2}$.  
Koshevoy \& Mosler \cite[Theorem 5.2]{KM} noticed that in the $d=1$ case, that
the lift zonoid order is exactly the convex order.   

We conclude this section by exploiting Theorem \ref{th:zonoid} to
obtain an interesting identity.  A similar formula can be found in the
paper of  Hiriart-Urruty  \& Mart\'{i}nez-Legaz  \cite{HM}.

\begin{theorem}\label{th:inverse}
Given $C \in \mathcal C$, let 
$$
\hat C^{-1}(q) = \inf\{ p \ge 0: \hat C(p) \ge q \} \mbox{ for all } 0 \le q \le 1.
$$
Then 
$$
\widehat{ C^*} (p ) = 1 - \hat C^{-1}(1-p) \mbox{ for all } 0 \le p \le 1.
$$
\end{theorem}

\begin{proof}
Let $S$ be a primal representation and $S^*$ be a dual representation of $C$.

Note that for all $0 \le p\le 1$ we have
$$
\hat C(p) - \hat C(0) = \sup\{ \EE[ S g(S)] : g:\RR \to [0,1]
 \mbox{ with }  \EE[g(S) ] = p \}
$$
and hence for any $0 \le q \le 1$ we have
\begin{align*}
\hat C^{-1}(q) & = \inf\{ \EE[ g(S) ], \  g:\RR \to [0,1]
 \mbox{ with } \EE[ S g(S) ] = q - \hat C(0) \} \\
& = 1- \sup\{ \EE[ g(S) ], \  g:\RR \to [0,1]
 \mbox{ with } \EE[ S g(S) ] = 1 -q \} \\
& = \PP(S>0) - \sup\{ \EE[ g(S) \one_{\{ S > 0\}} ], \  g:\RR \to [0,1]
 \mbox{ with } \EE[ S g(S) \one_{\{ S > 0\}} ] = 1 -q \} \\
& = \EE(S^*)  - \sup\left\{ \EE\left[ S^* g(S^*) \one_{\{ S^* > 0\}}  \right], \  g:\RR \to [0,1] \mbox{ with } 
\EE\left[   g(S^*) \one_{\{ S^* > 0\}}  \right] = 1- q  \right\} \\
& = 1 - \widehat{C^*}(1-q)
\end{align*} 
where we have used the observation that the optimal $g$ in the final maximisation problem
assigns zero weight to the event $\{ S^* = 0 \}$.
\end{proof}

\subsection{An extension}
Let $F$ be the distribution function of a log-concave density $f$ which is supported 
on all of $\RR$, so that $L=-\infty$ and $R = + \infty$ in the notation of section \ref{se:peacock}.
Let
$$
\hat C_f(p,y) = F( F^{-1}(p) + y ) \mbox{ for all } 0 \le p \le 1, y \in \RR.
$$
By Theorem \ref{th:first-duality} we have
$$
\hat C_f(p,y) = \widehat{C_f(\cdot, y)}(p) \mbox{ for all } 0 \le p \le 1, y \ge 0.
$$
It is interesting to note that the family of functions $(\hat C_f(\cdot ,y))_{y \in \RR}$ 
is a \textit{group} under function composition, not just a semigroup.  Indeed, we have
$$
\hat C_f(\cdot ,-y) = \hat C_f(\cdot ,y)^{-1}  \mbox{ for all } y \in \RR.
$$
Note that $\hat C_f(\cdot ,y)$ is increasing for all $y$, is concave if $y \ge 0$ but is
convex if $y < 0$.   In particular, when $y<0$ the function $\hat C_f(\cdot ,y)$ is
\textit{not} the concave conjugate of a call function in $\mathcal C$.
Unfortunately, the probabilistic or financial interpretation of the
inverse is not readily apparent.

For comparison, note that for $y \ge 0$ we have by Theorem \ref{th:inverse} that
\begin{align*}
\widehat{C_f(\cdot, -y)}(p) & = \widehat{C_f(\cdot, y)^*}(p) \\
& = 1 - F( F^{-1}(1-p) -y )  \mbox{ for all } 0 \le p \le 1. 
\end{align*}

\section{Acknowledgement} 
I would like to thank the Cambridge Endowment for Research in Finance for 
their support. I would also like to thank Thorsten Rheinl\"{a}nder for introducing 
me to the notion of a lift zonoid, and Monique Jeanblanc for introducing me to the notion of a lyrebird.   I would like to
thank the  participants of the London Mathematical Finance Seminar Series
and the Oberwolfach Workshop on the Mathematics of Quantitative Finance, where this work 
was presented.   After the original submission of this work, I learned that
Peter Carr and Greg Pelts \cite{Carr} independently proposed modelling call 
price curves
via their Legendre transform. I would like to thank Johannes Ruf for noticing
 this connection.
 
 I would also like to thank Johannes for a useful discussion 
 of the implication (5)$\Rightarrow$ (1) in Theorem \ref{th:char-surf}.  Originally, I had a complicated proof of this implication only in 
 the discrete-time case. The original argument was
 similar to the construction in the proof of Theorem \ref{th:bullet-S}: given the implication (5) $\Rightarrow$ (2),
 to apply the discrete-time It\^o--Watanabe decomposition to the supermartingale $S$, as in the construction
of F\"ollmer's exit measure. The difficulty in extending this argument to 
the continuous-time case  is that the local martingale appearing in the continuous-time It\^o--Watanabe decomposition may not be a true martingale.  While discussing this
technical point, Johannes inspired me to try to prove 
that, in fact, the stronger implication (5) $\Rightarrow$   (4) holds.  
 
  Finally, I would like to thank the referees for their useful comments on the content and presentation of this work. In particular, I thank them for 
  encouraging me to expand section \ref{se:calibrating}.

\end{document}